\documentclass[11pt]{article}
\usepackage{amsmath,amssymb,amsthm} 
\usepackage{bbm}
\usepackage{stmaryrd}
\usepackage{paralist}
\usepackage{hyperref}
\usepackage{color}
\usepackage{multirow}

\setdefaultenum{(i)}{}{}{}
\usepackage[margin=1in]{geometry}
\usepackage[pdftex]{graphicx}
\usepackage{subfigure}
\theoremstyle{plain}
\newtheorem{mythm}{Theorem} \numberwithin{mythm}{section}

\newtheorem{mylemma}[mythm]{Lemma}

\newtheorem{myrek}[mythm]{Remark}
\newtheorem{myex}[mythm]{Example}

\DeclareMathAlphabet\scr{U}{scr}{m}{n}
\SetMathAlphabet\scr{bold}{U}{scr}{b}{n}
  \DeclareFontFamily{U}{scr}{\skewchar\font'177}%
  \DeclareFontShape{U}{scr}{m}{n}{<-6>rsfs5<6-8>rsfs7<8->rsfs10}{}%
  \DeclareFontShape{U}{scr}{b}{n}{<-6>rsfs5<6-8>rsfs7<8->rsfs10}{}%

\numberwithin{equation}{section}

\renewenvironment{thebibliography}[1]{%
\begin{oldthebibliography}{#1}%
\setlength{\baselineskip}{.9em}
\linespread{.9}
\small
\setlength{\parskip}{0ex}%
\setlength{\itemsep}{.1em}%
}%
{%
\end{oldthebibliography}%
}

\DeclareMathOperator{\esssup}{ess\ sup}

\begin{document}
\title{Portfolio Choice with Stochastic Investment Opportunities: \\ a User's Guide\footnote{The authors thank Thomas Cay\'e, Paolo Guasoni, Martin Herdegen, Sebastian Herrmann, Marcel Nutz, Scott Robertson, and the participants of the 2013 Princeton RTG Summer School in Financial Mathematics for fruitful discussions.}}
\author{
Ren Liu\thanks{ETH Z\"urich, Departement Mathematik, R\"amistrasse 101, CH-8092, Z\"urich, Switzerland, email:
\texttt{ren.liu@math.ethz.ch}.}
\and
Johannes Muhle-Karbe\thanks{ETH Z\"urich, Departement Mathematik, R\"amistrasse 101, CH-8092, Z\"urich, Switzerland, email:
\texttt{johannes.muhle-karbe@math.ethz.ch}. Partially supported by the National Centre of Competence in Research ``Financial Valuation and Risk Management'' (NCCR FINRISK), Project D1 (Mathematical Methods in Financial Risk Management), of the Swiss National Science Foundation (SNF).}}

\date{November 7, 2013}
\pagestyle{plain}
\maketitle

\begin{abstract}
This survey reviews portfolio choice in settings where investment opportunities are stochastic due to, e.g.,  stochastic volatility or return predictability. It is explained how to heuristically \emph{compute} candidate optimal portfolios using tools from stochastic control, and how to rigorously \emph{verify} their optimality by means of convex duality. Special emphasis is placed on long-horizon asymptotics, that lead to particularly tractable results.
\end{abstract}




\section{Introduction}

Dynamic portfolio choice considers how to trade to maximize some performance criterion.  In continuous time, problems of this kind were first solved in the seminal works of Merton \cite{merton.69,merton.71}. In a Markovian setting, he used tools from stochastic  control to characterize optimal portfolios by means of partial differential equations. For investors with constant relative risk aversion and for a \emph{constant opportunity set}, i.e.,\ constant interest rates, expected excess returns, and volatilities, Merton found that it is optimal to hold a constant fraction of wealth in each risky asset. Perhaps surprisingly, these target weights are the same for all planning horizons.

Since Merton's groundbreaking results, portfolio choice has evolved in various directions. One strand of research has been the development of general existence and uniqueness results by means of the duality theory of convex analysis.\footnote{For complete markets, cf.\ Pliska \cite{pliska.86}, Karatzas, Lehoczky and Shreve \cite{karatzas.al.87}, as well as Cox and Huang \cite{cox.huang.89,cox.huang.91}. Extensions to the considerably more involved case of incomplete markets were studied by He and Pearson \cite{he.pearson.91a,he.pearson.91b}  as well as Karatzas, Lehoczky, Shreve and Xu \cite{karatzas.al.91}, culminating in the work of Kramkov and Schachermayer \cite{kramkov.schachermayer.99,kramkov.schachermayer.03}, where general semimartingale prices and arbitrary preferences are treated under minimal assumptions.} Other studies have focused on the effect of various market imperfections, such as trading constraints (e.g., \cite{cvitanic.karatzas.92,grossman.villa.92,grossman.zhou.93}), transaction costs (e.g., \cite{magill.constantinidis.76,constantinides.86,davis.norman.90,dumas.luciano.91,shreve.soner.94}), and taxes (e.g., \cite{dybvig.koo.96, dammon.al.01,ben.al.10}). In parallel, there has been a lot of work on understanding the effects of \emph{stochastic investment opportunity sets}, i.e., interest rates \cite{korn.kraft.01,munk.sorensen.04}, excess returns \cite{kim.omberg.96,barberis.00,wachter.02} and volatilities  \cite{chacko.viceira.05,benth.al.03,kraft.05,liu.07,kallsen.muhlekarbe.10} driven by some stochastic state variable. 

Numerous excellent textbooks treat portfolio choice with constant investment opportunities (cf., e.g., \cite{korn.97,fleming.soner.06,pham.09}). Here, we present such an introductory and self-contained -- yet rigorous -- treatment of more involved settings with stochastic investment opportunities.\footnote{More practical aspects are discussed in \cite{campbell.viceira.03}.}  In these models, portfolio choice focuses on determining \emph{intertemporal hedging terms}, through which investors take into account the future fluctuations of the investment opportunity set. These generally lead to time-inhomogeneous portfolios, that complicate both computation and interpretation. Hence, we place particular emphasis on \emph{long-run asymptotics} (cf.\ Guasoni and Robertson \cite{guasoni.robertson.12} and the references therein), where the horizon is postponed to infinity to produce simple time-homogenous policies that nevertheless perform well over reasonable finite horizons. Our goal is to introduce readers to a two-step procedure for solving such problems: first, candidates for the value function and the optimal portfolio are identified through partially heuristic arguments from stochastic control. Then, the optimality of these candidates is verified rigorously by means of convex duality. We illustrate this approach for the Heston-type stochastic volatility model studied in \cite{liu.07,kraft.05,kallsen.muhlekarbe.10} and the model with predictable returns analyzed by \cite{kim.omberg.96}. In addition to introducing the mathematical tools needed to rigorously analyze such models, we also discuss the economic significance of the results. The hedging terms arising due to stochastic volatility are typically small, as are the corresponding welfare effects \cite{chacko.viceira.05}. In contrast, predictable returns lead to substantial intertemporal hedging and large welfare gains~\cite{barberis.00}, at least in settings that ignore market frictions (compare \cite{lynch.tan.11}) and parameter uncertainty (see \cite{barberis.00}). These results appear to be well-known in the financial literature (cf., e.g., \cite{campbell.viceira.03} and the references therein), but do not seem to have percolated widely enough among financial mathematicians.

The remainder of this survey is organized as follows. The basic setting is introduced in Section~2. Subsequently, we discuss how to heuristically compute candidate optimal portfolios using the dynamic programming approach of stochastic optimal control. Finally, in Section~4, we explain how to rigorously verify the optimality of these candidates using convex duality.

\section{Setup}

\subsection{Financial Market}

On a filtered probability space $(\Omega,\scr{F},(\scr{F}_t)_{t \geq 0},\mathbb{P})$ supporting all the random quantities appearing in the sequel, we consider a frictionless financial market with two assets. The first one is \emph{safe}, in that it appreciates at the constant riskless rate $r>0$:
$$dS^0_t=r S^0_t dt, \quad S^0_0=1.$$
The second asset is \emph{risky}, with returns $dS_t/S_t$ of the following form:
$$\frac{dS_t}{S_t}=(\mu(Y_t)+r) dt+\sigma(Y_t) dW_t, \quad S_0=s>0.$$
Here, $W_t$ is a standard Brownian motion, and the \emph{excess return} $\mu(Y_t)$ as well as the \emph{volatility} $\sigma(Y_t)$ are deterministic functions of a \emph{state variable} $Y_t$, which follows an autonomous diffusion:
\begin{equation}\label{eq:dyny}
dY_t=b(Y_t)dt+a(Y_t)dW^Y_t, \quad Y_0=y.
\end{equation}
The process $W^Y_t$ is another Brownian motion, with constant correlation $\rho \in [-1,1]$ relative to $W_t$; the functions $\mu(y),\sigma(y), b(y), a(y)$ are assumed to be sufficiently regular to warrant the existence of a unique strong solution of the above stochastic differential equations.

In this survey, we illustrate how to tackle portfolio choice problems in this setting by a combination of tools from stochastic optimal control and convex duality. Particular emphasis is placed on the following three benchmark models that can be solved in closed form:

\begin{myex}\label{ex:models}
\begin{enumerate}
\item \emph{Black-Scholes Model}: In this simplest benchmark, both the expected excess return and the volatility are constant, $\mu(y)=\mu>0$ and $\sigma(y)=\sigma>0$. That is, there is no additional state variable and the \emph{investment opportunity set} $(r,\mu,\sigma)$ is constant.
\item \emph{Heston-type model}: This extension of the Black-Scholes model acknowledges that volatility is time varying and mean reverting. Modeling the infinitesimal variance process by Feller's square-root diffusion then leads to the following model:\footnote{In the context of option pricing, this model was proposed by Heston \cite{heston.93}. Under a risk-neutral measure, the drift rate of the risky asset is fixed to the riskless return $r$. For portfolio choice under the physical probability, one also needs to specify the risk premium $\mu(Y_t)$. The specification $\mu(Y_t)=\mu_S Y_t$ linear in variance is due to Liu \cite{liu.07}. It leads to particularly tractable results, because it corresponds to a Black-Scholes model after passing from usual calendar time to business time by means of a suitable random time change.}
\begin{align}
\frac{d S_t}{S_t} &= (\mu_S Y_t+r) dt + \sqrt{Y_t} d W_t,\label{eq:heston1}\\
d Y_t &= \lambda_Y (\bar Y- Y_t) dt + \sigma_Y \sqrt{Y_t} dW^Y_t.\label{eq:heston2}
\end{align}
Here, $\mu_S>0$ specifies the risk premium per unit variance, $\bar Y>0$ is the long-run mean of the variance process $Y_t$, $\sigma_Y>0$ governs the rate at which $Y_t$ fluctuates around $\bar Y$, and $\lambda_Y>0$ describes the speed at which $Y_t$ mean-reverts back to its long-run mean $\bar Y$. The two Brownian motions $W_t, W^Y_t$ have constant correlation $\rho \in [-1,0]$; empirical studies typically find substantially negative values (cf., e.g., \cite{pan.02}).\footnote{This is the so-called ``leverage effect'', i.e., the empirical observation that large downward moves of the asset price tend to be associated with upward moves of the volatility.} Finally, the Feller condition $2\lambda_Y \bar Y > \sigma_Y^2$ ensures that the variance process $Y_t$ remains strictly positive. 
\item \emph{Kim and Omberg Model}: A different extension of the Black-Scholes model proposed by Kim and Omberg \cite{kim.omberg.96} is to keep the volatility constant, while assuming that excess returns are \emph{predictable}.\footnote{Here, predictability refers to the ability to forecast returns using some auxiliary state variable, and should not be confused with the notion from stochastic calculus.} This means that they follow a mean-reverting process correlated with the Brownian motion $W_t$ driving the returns $dS_t/S_t$ of the risky asset:\footnote{Predictors that have been proposed in the empirical literature include stock characteristics, such as the dividend yield and earnings-price ratio, or interest rates, such as the term-spread and the corporate-spread.}
\begin{align}
\frac{d S_t}{S_t} &= (Y_t+r) dt + \sigma d W_t, \label{eq:ko1}\\
d Y_t &= \lambda_Y (\bar Y- Y_t) dt + \sigma_Y d W_t^{Y}. \label{eq:ko2}
\end{align}
Here, the excess return follows an Ornstein-Uhlenbeck process with long-run mean $\bar Y>0$, volatility $\sigma_Y>0$, and mean-reversion speed $\lambda_Y>0$. The Brownian motions $W_t$ and $W^Y_t$ have constant correlation $\varrho \in [-1,0]$; empirical studies suggest values close to $-1$ (cf.\ \cite{barberis.00}).
\end{enumerate}
\end{myex}

Fix an \emph{initial endowment} $x>0$. \emph{Trading strategies} are modeled by continuous processes $\pi_t$ integrable with respect to the return process $dS_t/S_t$; here, $\pi_t$ denotes the fraction of wealth invested in the risky asset at time $t$. The returns of the associated wealth process $X^\pi_t$ are then given by the corresponding convex combination of the safe and risky returns: 
\begin{equation}\label{eq:dyn}
dX^\pi_t/X^\pi_t=(1-\pi_t)dS^0_t/S^0_t+\pi_t dS_t/S_t = (r+\pi_t\mu(Y_t))dt +\pi_t \sigma(Y_t) dW_t, \quad X^\pi_0=x.
\end{equation}
As the solution of this linear stochastic differential equation, the wealth process corresponding to the risky fraction $\pi_t$ is therefore given by the stochastic exponential
\begin{equation}\label{eq:stoexp}
X^\pi_t=x\scr{E}\left(\int_0^\cdot (r+\pi_u\mu(Y_u)) du +\int_0^\cdot \pi_u \sigma(Y_u) dW_u\right)_t.
\end{equation}
Note that the parametrization of trading strategies in terms of risky weights $\pi_t$ rules out doubling strategies, because it automatically leads to wealth processes that remain strictly positive.

\subsection{Preferences}

We consider an investor with preferences described by a \emph{power utility function} $U(x)=x^{1-\gamma}/(1-\gamma)$ with \emph{constant relative risk aversion} $0<-xU''(x)/U'(x)=\gamma \neq 1$. Given some planning horizon $T>0$, the investor's goal is to maximize expected utility from terminal wealth:
\begin{equation}\label{eq:tw}
E[U(X^\pi_T)] \to \max!,
\end{equation}
where $X^\pi_T$ runs through the set of terminal payoffs generated by portfolios $\pi_t$.

\section{Computation}

\subsection{The Hamilton-Jacobi-Bellman Equation}\label{sec:hjb}

In this section, we use methods from the theory of stochastic optimal control to \emph{heuristically} derive a candidate optimal strategy for the portfolio choice problem \eqref{eq:tw}.  The discussion is kept on an informal level, since we later turn to tools from convex duality for verification in Section \ref{sec:verification}. For a full exposition of the theory of stochastic control and complete proofs see, e.g., \cite{fleming.soner.06,korn.97,pham.09} and the references therein.

In the portfolio choice problem \eqref{eq:tw}, the initial endowment $x$, the initial value $y$ of the state variable, and the time to maturity $T$ are fixed. The key insight for the dynamic programming approach of stochastic optimal control is that it is easier to solve a dynamic version of this problem, where all of these variables are allowed to vary. To this end, consider the so-called \emph{value function}, which describes the maximal utility that can be attained by trading optimally on $[t,T]$ starting from some wealth $x$ and value $y$ of the state variable:
\begin{equation}\label{eq:vf}
V(t,x,y)= \esssup_{(\pi_u)_{u \in [t,T]}} E\left[ U\left(X^\pi_T\right)\Big|X^\pi_t=x, Y_t=y\right],
\end{equation}
where $(\pi_u)_{u\in [t,T]}$ runs through all portfolios on $[t,T]$.\footnote{Note that the value function does not depend on the current values $S^0_t,S_t$ of the safe and risky assets, because these do not feature in the dynamics \eqref{eq:dyn} and \eqref{eq:dyny} of the wealth process and the state variable, respectively, for \emph{feedback controls} $\pi_t=\pi(t,X^\pi_t,Y_t)$. The general rule of thumb is that one needs to add enough variables so that their concatenation with the wealth process forms a Markov process. Then, the optimal control and in turn the conditional expectation \eqref{eq:vf} should indeed be deterministic functions of the initial states.}  Then, evidently, the value function satisfies the terminal condition 
\begin{equation}\label{eq:terminal1}
V(T,x,y)=U(x).
\end{equation} 
The key step now is to derive a partial differential equation that describes its evolution before the terminal time, and to use the latter to characterize the optimal policy. Here, we follow the heuristic argument presented by Pham \cite{pham.09}. The starting point is the \emph{dynamic programming principle of stochastic control}, which states that
\begin{equation}\label{eq:dpp}
V(t,x,y) \geq  E\left[V(t+h,X^{\pi}_{t+h},Y_{t+h}) | X_t^\pi=x, Y_t=y\right],
\end{equation}
for $h \in [0,T-t]$ and all portfolios $\pi_t$. This means that following any policy on $[t,t+h]$ and then switching to the optimizer on $[t+h,T]$ cannot -- on average -- produce better results than holding the optimal portfolio for $[t,T]$ on this entire time interval. This implies that the optimization problem on $[t,T]$ can be solved recursively: taking an arbitrary starting value $X^\pi_{t+h}$ as given, one first determines the optimal control on the interval $[t+h,T]$ and the corresponding value function $V(t+h,X^{\pi}_{t+h},Y_{t+h})$. Then, one searches for the optimal control on $[t,t+h]$ by maximizing $E[V(t+h,X^\pi_{t+h},Y_{t+h})| X_t^\pi=x,Y_t = y]$ over all controls on $[t,t+h]$. Pasting together the two portfolios then should lead to the optimal strategy. 

By applying the dynamic programming principle on infinitesimally short intervals, one can derive a partial differential equation for the value function. Indeed, consider a constant control $\pi$ applied over a small time interval $[t,t+h]$.  Assuming that the value function is smooth enough, It\^o's formula yields 
\begin{align}
&V(t+h,X^\pi_{t+h},Y_{t+h})= V(t,X^\pi_t,Y_t) \label{eq:dynv}\\
&\quad+\int_t^{t+h} \left(V_t+V_x X^\pi_u r+ V_x X^\pi_u \pi \mu + V_y b+\frac{1}{2} V_{xx}  (X^\pi_u)^2 \pi^2 \sigma^2 + V_{xy} X^\pi_u \pi \rho \sigma a + \frac{1}{2} V_{yy}  a^2\right)du \notag\\
&\quad+\int_t^{t+h} V_x X^\pi_u \pi \sigma dW_u +\int_t^{t+h} V_y a dW^Y_u,\notag
\end{align}
where the arguments of the functions are omitted for brevity. Inserting this into \eqref{eq:dpp}, we obtain
\begin{align*}
0 \geq E\Bigg[\int_t^{t+h} \Big( V_t &+V_x X^\pi_u r+ V_x X^\pi_u \pi \mu + V_y b\\
&+\frac{1}{2} V_{xx}  (X^\pi_u)^2 \pi^2 \sigma^2 + V_{xy} X^\pi_u \pi \rho \sigma a + \frac{1}{2} V_{yy}  a^2 \Big) du \Big| X^\pi_t=x, Y_t=y\Bigg],
\end{align*}
assuming that the stochastic integrals with respect to the Brownian motions $W_t, W^Y_t$ are true and not only local martingales, and therefore vanish in expectation. Now, divide by $h$ and then let $h$ tend to zero, obtaining 
\begin{equation*}
0 \geq V_t+V_x x r+ V_x x \pi \mu + V_y b+\frac{1}{2} V_{xx}  x^2 \pi^2 \sigma^2 + V_{xy} x \pi \rho \sigma a + \frac{1}{2} V_{yy}  a^2.
\end{equation*}
Since this holds for \emph{any} constant control $\pi$, it follows that
\begin{equation}\label{eq:hjbpre}
0 \geq V_t +\sup_\pi\left\{V_x x r+ V_x x \pi \mu + V_y b+\frac{1}{2} V_{xx}  x^2 \pi^2 \sigma^2 + V_{xy} x \pi \rho \sigma a + \frac{1}{2} V_{yy}  a^2\right\}.
\end{equation}
On the other hand, \eqref{eq:dpp} should hold with equality for the \emph{optimal} control $\widehat\pi_t$. Repeating the derivation of \eqref{eq:hjbpre}, this suggests that this equation actually holds with equality, where the pointwise supremum is attained for the optimal portfolio $\widehat\pi_t$:
\begin{equation}\label{eq:HJB1a}
0 = V_t +\sup_\pi\left\{V_x x r+ V_x x \pi \mu + V_y b+\frac{1}{2} V_{xx}  x^2 \pi^2 \sigma^2 + V_{xy} x \pi \rho \sigma a + \frac{1}{2} V_{yy}  a^2\right\}.
\end{equation}
This PDE is called the \emph{dynamic programming equation} or \emph{Hamilton-Jacobi-Bellman} (henceforth \emph{HJB}) \emph{equation}. In the present setting with only one risky asset, the candidate optimal portfolio is readily determined by pointwise optimization as a function of time, wealth, and the state variable: 
\begin{equation}\label{eq:cand1}
\widehat{\pi}(t,x,y)=-\frac{V_{x}(t,x,y)}{x V_{xx}(t,x,y)} \frac{\mu(y)}{\sigma(y)^2} -\frac{V_{xy}(t,x,y)}{x V_{xx}(t,x,y)} \frac{\rho a(y)}{\sigma(y)}.
\end{equation}
After inserting this expression, the HJB equation \eqref{eq:HJB1a} for the value function $V(t,x,y)$ reads as follows:
\begin{equation}\label{eq:HJB1}
V_t=\frac{1}{2} \frac{V_x^2}{V_{xx}} \frac{\mu}{\sigma^2} +\frac{V_{x}V_{xy}}{V_{xx}} \frac{\mu\rho a}{\sigma}+\frac{1}{2}\frac{V_{xy}^2}{V_{xx}}\rho^2 a^2-V_x x r-V_y b -\frac{1}{2} V_{yy} a^2.
\end{equation}

\subsection{Homotheticity}

The above heuristics were valid for an arbitrary utility function. Power utilities $U(x)=x^{1-\gamma}/(1-\gamma)$, however, are particularly tractable because they distinguish themselves through their \emph{homotheticity}, $U(x)=x^{1-\gamma} U(1)$. Since the wealth process of any portfolio $\pi_t$ is proportional to the initial endowment $x$ by \eqref{eq:stoexp}, this property allows to factor out wealth from the value function:
\begin{align}\label{eq:hom}
V(t,x,y)=&\esssup_{(\pi_u)_{u \in [t,T]}} E\left[U(X^\pi_T) | X^\pi_t=x, Y_t=y\right] \notag\\
= &x^{1-\gamma} \esssup_{(\pi_u)_{u \in [t,T]}}  E\left[ U(X^\pi_T) | X^\pi_t=1, Y_t=y\right]=x^{1-\gamma} V(t,1,y),
\end{align}
Define the \emph{reduced value function} 
$$v(t,y)=(1-\gamma)V(t,1,y).$$
Then, the candidate optimal portfolio \eqref{eq:cand1} can be written as
\begin{equation}\label{eq:cand2}
\widehat{\pi}_t=\frac{\mu}{\gamma \sigma^2} +\frac{\rho a}{\gamma \sigma} \frac{v_y}{v},
\end{equation}
and the corresponding \emph{reduced HJB equation} reads as
\begin{equation}\label{eq:HJB2}
v_t=\frac{\gamma-1}{\gamma}\left(\left(\frac{\mu^2}{2\sigma^2}+\gamma r\right)v+\frac{\mu\rho a}{\sigma}v_y+\frac{\rho^2 a^2}{2} \frac{v_y^2}{v}\right)-v_y b -\frac{1}{2} v_{yy} a^2.
\end{equation}
In view of the terminal condition \eqref{eq:terminal1} for the original value function and the homotheticity \eqref{eq:hom}, the corresponding terminal condition is 
\begin{equation}\label{eq:terminal2}
v(T,y)=1.
\end{equation}
The PDE \eqref{eq:HJB2} can be linearized by a suitable power transformation \cite{zariphopoulou.01}. This allows to prove existence and uniqueness in a general setting given sufficient regularity of the coefficients $\mu(y),\sigma(y),b(y)$, and $a(y)$. However, explicit solutions are generally not available. To shed more light on the various effects that can arise in dynamic portfolio choice, from now on we therefore focus our attention on the three benchmark models from Example \ref{ex:models}, which can be solved in closed form.

\subsection{Constant Investment Opportunities}\label{sec:BS1}

First, consider the Black-Scholes model with constant expected returns and volatilities. In this simplest example, there is no extra state variable $Y_t$, so that the value function only depends on time and current wealth. Accordingly, all the partial $y$-derivatives vanish, which implies that the candidate \eqref{eq:cand2} for the optimal portfolio is constant:
$$\widehat{\pi}=\frac{\mu}{\gamma\sigma^2}.$$
The HJB equation \eqref{eq:HJB2} for the corresponding candidate value function reduces to an ODE:
$$v'(t)=-(1-\gamma)\left(\frac{\mu^2}{2\gamma\sigma^2}+r\right)v(t).$$
Together with the terminal condition \eqref{eq:terminal2} and the homotheticity \eqref{eq:hom}, this leads to 
$$V(t,x)=\frac{x^{1-\gamma}}{1-\gamma}\exp\left((1-\gamma)\left(r+\frac{\mu^2}{2\gamma\sigma^2}\right)(T-t)\right).$$
\paragraph{Discussion}
The candidate optimal policy in the Black-Scholes model is to hold a constant proportion $\widehat\pi$ of wealth in the risky asset, irrespective of the investment horizon. This constant weight is given by the infinitesimal \emph{mean-variance ratio} $\mu/\sigma^2$ of the risky excess return, divided by risk aversion $\gamma$. The portfolio $\widehat\pi$ is \emph{myopic}, in that it is fully determined by the local dynamics of the return process, but ignores how much of the investment period remains left. 

The corresponding utility is equivalent to the one obtained from fully investing at a fictitious \emph{equivalent safe rate} $r+\frac{\mu^2}{2\gamma\sigma^2}$. Here, the outperformance of the actual safe rate is determined by the square of the instantaneous \emph{Sharpe ratio} $\mu/\sigma$ of the risky asset, scaled by risk aversion. Hence, with constant expected excess returns and volatilities, investors with constant relative risk aversion use the Sharpe ratio to rank the attractiveness of different risky assets.

\subsection{Stochastic Volatility}\label{sec:heston1}

In the Heston-type model (\ref{eq:heston1}-\ref{eq:heston2}), there is a nontrivial state variable $Y_t$, so that the additional partial $y$-derivatives in the HJB equation \eqref{eq:HJB2} come into play. After inserting the dynamics (\ref{eq:heston1}-\ref{eq:heston2}), one finds that all coefficients in the PDE \eqref{eq:HJB2} are affine linear functions of the state variable. This suggests the following exponentially affine ansatz for the reduced value function:
\begin{equation}\label{eq:affine}
v(t,y)=\exp\left(A(t)+B(t)y\right),
\end{equation}
for smooth functions $A(t),B(t)$ satisfying $A(T)=B(T)=0$ to match the terminal condition \eqref{eq:terminal2}. With \eqref{eq:affine} as well as (\ref{eq:heston1}-\ref{eq:heston2}), the HJB equation \eqref{eq:HJB2} becomes
\begin{align*}
A'(t)+B'(t) y= \frac{\gamma-1}{\gamma}\left(\left(\frac{\mu_S^2y}{2}+\gamma r\right)+\mu_S \rho \sigma_Y y B(t)+\frac{\rho^2\sigma_Y^2}{2}B(t)^2 \right)-B(t)\lambda_Y(\bar {Y}-y)-\frac{\sigma_Y^2}{2}B(t)^2,
\end{align*}
after canceling the common  factor $v(t,y)$. This equation should hold for all values of $y$. Hence, separating the terms proportional to resp.\ independent of $y$ leads to the following system of ordinary differential equations for the functions $B(t), A(t)$:
\begin{align}
B'(t) &=cB(t)^2+bB(t)+a, \quad B(T) = 0,\label{eq:ode1}\\
A'(t) &= (\gamma-1)r-\lambda \bar{Y} B(t), \quad A(T)=0,\label{eq:ode2}
\end{align}
where
\begin{equation}\label{eq:abcheston}
c=\left(\frac{\gamma-1}{\gamma}\rho^2-1\right)\frac{\sigma_Y^2}{2}, \quad b=\left(\frac{\gamma-1}{\gamma}\mu_S \rho \sigma_Y+\lambda_Y\right), \quad a=\frac{\gamma-1}{\gamma}\frac{\mu_S^2}{2}.
\end{equation}
The first ODE \eqref{eq:ode1} is a Riccati equation for $B(t)$, whose solution for various parameter constellations can be readily determined from integral tables (e.g., \cite[21.5.1.2]{bronstein.01}).\footnote{For low risk aversion $\gamma<1$, some parameter restrictions are needed to guarantee that a nonexplosive solution exists on the whole interval $[0,T]$, compare \cite[Section 3.1]{kallsen.muhlekarbe.10}.} If the discriminant 
\begin{equation}\label{eq:Dheston}
D=b^2-4ac
\end{equation}
is positive,\footnote{This always holds for high risk aversion $\gamma>1$. For low risk aversion $\gamma<1$, this condition is satisfied, e.g., if the volatility $\sigma_Y$ of the variance process $Y_t$ is sufficiently small compared to its mean-reversion speed $\lambda_Y$, i.e., if the stochastic volatility does not fluctuate too widely around its long-run mean.\label{discrimantpos}} it is given by \cite{kraft.05,liu.07,kallsen.muhlekarbe.10}:
\begin{equation}\label{eq:Bheston}
B(t)=-2a\frac{e^{\sqrt{D}(T-t)}-1}{e^{\sqrt{D}(T-t)}(b+\sqrt{D})-b+\sqrt{D}}.
\end{equation}
The ODE \eqref{eq:ode2} in turn allows to determine the function $A(t)$ by a simple integration:
\begin{align}
A(t)&=(1-\gamma)r(T-t)+\lambda_Y \bar{Y} \int_t^T B(u) du \label{eq:Aheston}\\
&=(1-\gamma)r(T-t)-\frac{2\lambda_Y \bar{Y} a}{b^2-D}\left((b+\sqrt{D})(T-t)-2\log\left(\frac{e^{\sqrt{D}(T-t)}(b+\sqrt{D})-b+\sqrt{D}}{2\sqrt{D}}\right)\right).\notag
\end{align}
These explicit formulas determine the candidate value function $V(t,x,y)=\frac{x^{1-\gamma}}{1-\gamma} \exp(A(t)+B(t)y)$. In view of \eqref{eq:cand2}, the candidate for the corresponding optimal portfolio is given by
\begin{equation}\label{eq:portfolio_heston_finite}
\widehat\pi_t=\frac{\mu_S}{\gamma}+\frac{\rho\sigma_Y}{\gamma}B(t).
\end{equation}

\paragraph{Discussion}
Despite the stochastic opportunity set $(r, \mu_S Y_t, \sqrt{Y}_t)$ of the Heston model, the optimal risky weight is still deterministic. However, it is no longer constant like in the Black-Scholes model. In addition to the myopic part given by the constant infinitesimal mean-variance ratio $\mu_S$ of the risky excess returns, scaled by risk aversion $\gamma$, there is an additional \emph{intertemporal hedging term}, $\frac{\rho\sigma_Y}{\gamma} B(t)$, which vanishes as the horizon $T$ nears by the terminal condition $B(T)=0$.\footnote{As pointed out by Mossin~\cite{mossin.68}, an investor behaves \emph{myopically}, if her decisions are obtained by solving a series of single-period problems, i.e., decisions are made for one period without looking ahead. This always applies for the constant investment opportunity sets discussed in the previous section, where the optimal portfolio was the same for all horizons. Here, the myopic component is the investment for a very short horizon ($T \to 0$), that is now complemented by an additional hedging term for longer maturities $T$.} Further from maturity, however, it provides a hedge against the future evolution of the state variable $Y_t$. If the latter is uncorrelated with the Brownian motion $W_t$ driving the risky returns ($\rho=0$), then no hedging occurs. For a nontrivial correlation, the direction of the hedge depends on the investor's risk aversion. In the limit $\gamma \to 1$ (which corresponds to logarithmic utility $U(x)=\log(x))$ the hedging term again vanishes, in line with the general results for log-optimal portfolios.\footnote{For logarithmic utility, no intertemporal hedging occurs even if asset dynamics are non-Markovian and include jumps, see, e.g.\ \cite{karatzas.al.91,goll.kallsen.00}. That is, logarithmic investors always behave myopically.} For risk aversion $\gamma>1$, the ODE for $B(t)$ and a comparison argument readily show $B(t)<0$; conversely, $B(t) \geq 0$ for $\gamma <1$. Since $\rho<0$ as suggested by the empirical literature (e.g., \cite{pan.02}), this implies that investors with risk aversion above unity react to uncertainty about future volatility by holding larger proportions of the risky asset. In contrast, investors with risk aversion less than one decrease their risky weights, and may in fact initially hold negative positions in the risky asset even though the latter has a positive expected excess return \cite{liu.07}. 

These initially puzzling results can be understood as follows. In the specification \eqref{eq:heston1}, the infinitesimal Sharpe ratio of the risky asset is given by $\mu_S \sqrt{Y_t}$. Hence, it is increasing in volatility and therefore negatively correlated with the Brownian motion $W_t$ driving the risky returns. As a result, good returns today tend to occur together with low Sharpe ratios, and therefore bad investment opportunities,\footnote{This identification makes sense at least if the model is not too far from Black-Scholes (where performance is indeed measured by the Sharpe ratio) or risk aversion is not too far from unity (for log-utility, performance is also measured by the average squared Sharpe ratio even in general models).} tomorrow. Conversely, low returns today tend to occur with high Sharpe ratios and therefore good investment opportunities, tomorrow. A more risk averse investor ($\gamma>1$) will invest in this hedging opportunity by purchasing extra shares of the risky asset. Conversely, less risk averse investors add a negative ``hedging portfolio'' to the myopic component. For the latter, good returns today are correlated with good investment opportunities tomorrow, thereby allowing less risk averse investor to speculate on this scenario.  

These qualitative features delicately depend on the concrete model under consideration. For example, Chacko and Viceira \cite{chacko.viceira.05} consider a stochastic volatility model of the form
$$\frac{dS_t}{S_t}=(\mu_{CV}+r) dt +\sqrt{1/Y_t}dW_t,$$
for a constant expected excess return $\mu_{CV}>0$, and where the \emph{inverse} $Y_t$ of the variance process follows a square-root process. The state variable $Y_t$ is then positively correlated with $W_t$, to ensure a negative correlation between asset and volatility shocks. As a result, the model's instantaneous Sharpe ratio is given by $\mu_{CV} \sqrt{Y_t}$, and is therefore positively correlated with current asset returns. Consequently, the intertemporal hedging terms turn out to have the opposite signs compared to the specification of Liu \cite{liu.07} considered above. More risk averse investors then purchase fewer risky shares due to the clustering of bad current and future investment opportunities. In contrast, less risk averse investors speculate on this by holding larger risky positions \cite{chacko.viceira.05}.

\begin{figure}
\centering
\includegraphics[width=0.8\textwidth]{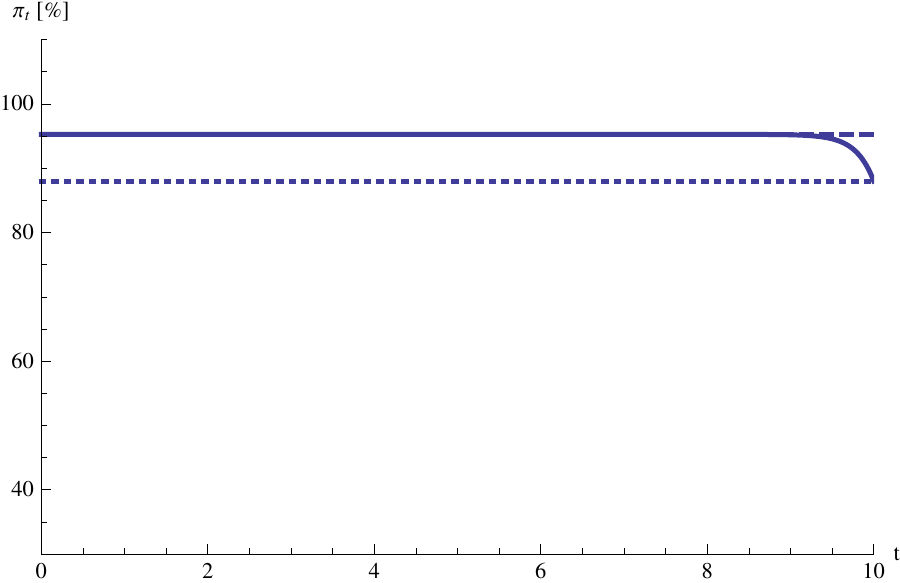}
\caption{\label{fig:portfolio_heston}
Heston-type model: optimal portfolio $\widehat\pi_t$ for horizon $T=10$ years (solid), long-run optimal portfolio $\widehat\pi^\infty$ (dashed), and Black-Scholes portfolio $\mu_S/\gamma$ for same mean return and volatility (dotted). Risk aversion is $\gamma=5$ and the (yearly)  model parameters are $r=0.033$, $\mu_S=4.4$, $\lambda_Y=5.3$, $\bar Y=0.024$, $\sigma_Y=0.38$ and $\rho=-0.57$ (cf.\ \cite[Tables 1 and 6]{pan.02}).
}
\end{figure}

In view of this lack of robustness, it is reassuring to notice that the magnitude of these effects does not turn out to be very large for typical parameter estimates. This is illustrated in Figure~\ref{fig:portfolio_heston}. For the parameter estimates of Pan~\cite{pan.02}, the size of the intertemporal hedging term turns out to be one order of magnitude smaller than the myopic component $\mu_S/\gamma$, which corresponds to the risky weight in a Black-Scholes model with the same mean returns $\mu_S \bar Y$ and variance $\bar Y$. Similar quantitative results are reported by \cite{campbell.viceira.03,chacko.viceira.05}.

\subsection{Predictable Returns}\label{sec:ko}

Now, let us turn to the model (\ref{eq:ko1}-\ref{eq:ko2}) with predictable returns proposed by Kim and Omberg~\cite{kim.omberg.96}. With the dynamics (\ref{eq:ko1}-\ref{eq:ko2}), the coefficients in the reduced HJB equation \eqref{eq:HJB2} are no longer affine linear in the state variable like for the Heston model, but in fact quadratic. This suggests the following exponentially quadratic ansatz for the reduced value function:
\begin{equation}\label{eq:quadratic}
v(t,y)=\exp\left(A(t)+B(t)y+\frac{1}{2}C(t)y^2\right),
\end{equation}
for smooth functions $A(t), B(t), C(t)$ satisfying $A(T)=B(T)=C(T)=0$ to match the terminal condition \eqref{eq:terminal2}. With \eqref{eq:quadratic} and (\ref{eq:ko1}-\ref{eq:ko2}), the reduced HJB equation \eqref{eq:HJB2} can be rewritten as:
\begin{align*}
A'(t)+B'(t)y+\frac{1}{2} C'(t)y^2=&\frac{\gamma-1}{\gamma}\left(\left(\frac{y^2}{2\sigma^2}+\gamma r\right)+\frac{y\rho\sigma_Y}{\sigma}(B(t)+C(t)y)+\frac{\rho^2\sigma_Y^2}{2}(B(t)+C(t)y)^2\right)\\
&\quad -\lambda(\bar Y-y)(B(t)+C(t)y)-\frac{\sigma_Y^2}{2}(C(t)+(B(t)+C(t)y)^2).
\end{align*}
This equation again has to hold for all values $y$ of the state variable. Hence, separating into terms proportional to $y^2$, $y$, and independent of $y$, this leads to the following system of ODEs for the three functions $C(t), B(t), A(t)$:
\begin{align*}
\frac{C'(t)}{2} &= c C(t)^2+bC(t)+a, \quad C(T)=0,\\
B'(t) &= 2c B(t)C(t)+bB(t)-\lambda_Y \bar{Y} C(t), \quad B(T)=0,\\
A'(t) &= (\gamma-1)r+cB(t)^2-\lambda_Y \bar{Y} B(t)-\frac{\sigma_Y^2}{2} C(t), \quad C(T)=0,
\end{align*}
where
\begin{equation}\label{eq:abcko}
c=\left(\frac{\gamma-1}{\gamma}\rho^2-1\right)\frac{\sigma_Y^2}{2}, \quad b=\frac{\gamma-1}{\gamma} \frac{\rho\sigma_Y}{\sigma}+\lambda, \quad a=\frac{\gamma-1}{\gamma} \frac{1}{2\sigma^2}.
\end{equation}
The first of the above ODEs is a Riccati differential equation for the function $C(t)$. The various parametric forms of its solution can be found in integral tables (cf., e.g., \cite[21.5.1.2]{bronstein.01}).\footnote{As for the Heston-type model, some parameter restrictions are needed for $\gamma<1$ to guarantee that a nonexplosive solution exists on the whole interval $[0,T]$; see Kim and Omberg \cite{kim.omberg.96} for a thorough discussion.} Given $C(t)$, the second equation is an inhomogeneous linear ODE for $B(t)$, that can be solved by variation of constants. Finally, with $C(t)$ and $B(t)$ at hand, $A(t)$ is obtained from the third equation by a simple integration. After rather tedious but straightforward calculations, this leads to the explicit formulas reported by Kim and Omberg \cite{kim.omberg.96}.

With the candidate value function at hand, \eqref{eq:cand2} in turn determines the corresponding candidate for the optimal portfolio as
$$\widehat\pi_t=\frac{Y_t}{\gamma \sigma^2}+\frac{\rho\sigma_Y}{\gamma\sigma}(B(t)+C(t)Y_t).$$

\paragraph{Discussion}
In contrast to the Heston-type model considered above, the candidate optimal portfolio for the model of Kim and Omberg (\ref{eq:ko1}-\ref{eq:ko2}) depends on the current value of the state variable $Y_t$. This applies both to the myopic component $\frac{Y_t}{\gamma\sigma^2}$ and to the intertemporal hedging term $\frac{\rho\sigma_Y}{\gamma\sigma}(B(t)+C(t)Y_t)$. 

In general, the sign of the hedging term therefore of course also depends on the state variable $Y_t$ (see \cite{kim.omberg.96} for a detailed discussion). For simplicity, let us consider the typical case where $Y_t$ is close to its long-run mean $\bar{Y}$ and therefore positive. One readily infers from the respective ODEs that $B(t)$ and $C(t)$ are negative for high risk aversion ($\gamma>1$) and positive for low risk aversions ($\gamma<1$). Since correlation $\rho$ is negative, in line with the empirical results of Barberis \cite{barberis.00} for a state variable representing the dividend yield, this implies that predictable returns increase risky investments if risk aversion $\gamma$ is higher than unity, and conversely for $\gamma<1$. This can again be understood by considering the correlation between current returns and future investment opportunities, measured by the model's infinitesimal Sharpe ratio. The latter is given by $Y_t/\sigma$, and is therefore negatively correlated with the Brownian motion $W_t$ driving the asset returns. Investors with high risk aversion therefore buy extra risky shares to hedge against bad returns tomorrow, whereas investors with low risk aversion speculate on this scenario by buying fewer risky shares.

With return predictability, the size of these hedging terms is considerably bigger than for the stochastic volatility model of Section \ref{sec:heston1}. Indeed, Figure \ref{fig:portfolio_ko} illustrates that they can easily become as big as the myopic part of the portfolio for typical parameter values.

\begin{figure}
\centering
\includegraphics[width=0.8\textwidth]{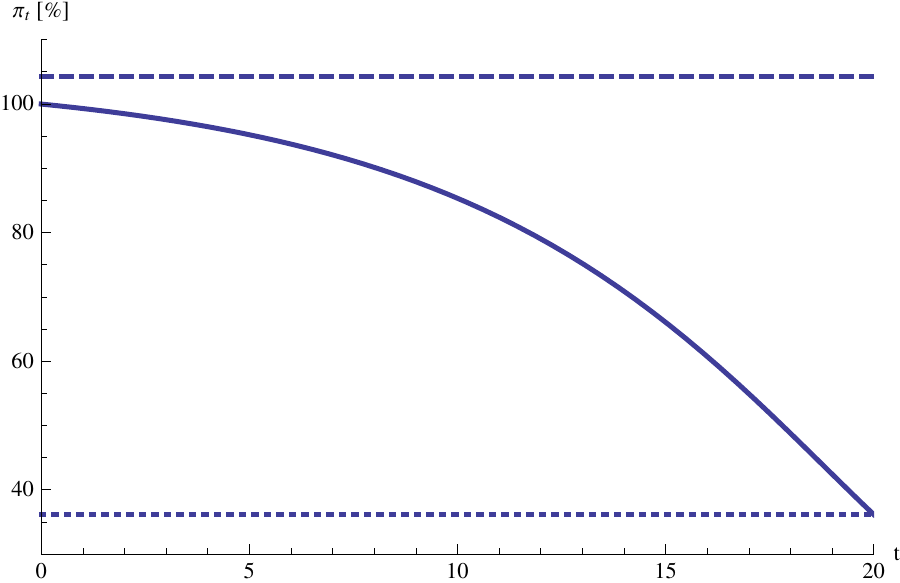}
\caption{\label{fig:portfolio_ko}
Kim and Omberg model: expected optimal portfolio $\widehat\pi_t$ for horizon $T=20$ years (solid), long-run optimal portfolio $\widehat\pi^\infty$ (dashed), and Black-Scholes portfolio $\bar{Y}/\gamma \sigma^2$ for same mean return and volatility (dotted). Risk aversion is $\gamma=5$, the (monthly) model parameters are $r=0.14\%$, $\sigma=4.36\%$, $\bar{Y}=0.34\%$, $\lambda_Y=2.26\%$, $\sigma_Y=0.08\%$, $\rho=-93.5\%$ (cf.~\cite{barberis.00,wachter.02}), and the state variable is at its long-run mean, $Y_t=\bar{Y}$.
}
\end{figure}

\subsection{Long-Run Asymptotics}\label{sec:lr1}

Unlike for the Black-Scholes model, the optimal portfolios in the models of Heston (\ref{eq:heston1}-\ref{eq:heston2}) as well as Kim and Omberg (\ref{eq:ko1}-\ref{eq:ko2}) depend on the investment horizon $T$. However, as the horizon grows, they converge to a stationary level, as illustrated in Figures \ref{fig:portfolio_heston} and \ref{fig:portfolio_ko}. This suggests that using simple temporally homogeneous portfolios should achieve ``almost'' optimal performances in both cases if the investment horizon is ``sufficiently long''. 

\paragraph{Long-Run Optimality} To make this precise, the terminal wealth problem \eqref{eq:tw} needs to be replaced by a stationary objective where the horizon is postponed to infinity. One approach is to consider utility from consumption over an infinite horizon. This works well in the context of proportional transaction costs (e.g., \cite{magill.constantinidis.76,davis.norman.90,shreve.soner.94}), but does not lead to tractable results for stochastic volatility models or predictable returns unless markets are complete (cf.\ Wachter \cite{wachter.02}). As a remedy, one can turn to the maximization of the \emph{long-run growth rate}
\begin{equation}\label{eq:lr}
\liminf_{T \to \infty} \frac{1}{T} \log U^{-1}(E[U(X^\pi_T)]) \to \max!,
\end{equation}
which also leads to tractable results for problems with proportional transaction costs \cite{dumas.luciano.91} or constraints on leverage \cite{grossman.villa.92} or drawdowns \cite{grossman.zhou.93}. This quantity corresponds to a fictitious \emph{equivalent safe rate}, at which -- in the long-run -- a full safe investment yields the same expected utility as trading optimally in the original market. Indeed, suppose the maximal growth rate in \eqref{eq:lr} is given by some $\beta$, and is attained by some portfolio $\widehat\pi_t$. Then, for large $T$:
$$E[U(X^{\widehat\pi}_T)] \approx E[U(xe^{\beta T})].$$
In the Black-Scholes model, the same portfolio is optimal for \emph{all} horizons $T$. Hence, this policy evidently also maximizes the long-run growth rate \eqref{eq:lr}. 

\paragraph{Stochastic Volatility} For the Heston model (\ref{eq:heston1}-\ref{eq:heston2}), an inspection of the explicit formulas (\ref{eq:Bheston}-\ref{eq:Aheston}) shows that the function $B(t)$ converges to a constant $B^\infty$ as the time to maturity $T-t$ grows. Likewise, the function $A(t)$ is almost linear for long horizons, $A(t) \approx A^\infty (T-t)$ for some constant $A^\infty$. Accordingly, we expect the long-run value function for a distant horizon to be of the following form:\footnote{Of course, this ansatz  does not solve any of the finite horizon problems because it does not satisfy the terminal condition $V(T,x,y)=U(x)$.}
\begin{equation}\label{eq:vlr}
V^\infty(t,x,y)=\frac{x^{1-\gamma}}{1-\gamma} \exp\left(A^\infty(T-t)+B^\infty y\right).
\end{equation}
In the limit for long horizons $T \to \infty$, the investor's equivalent safe rate then should be given by $A^\infty/(1-\gamma)$, as the importance of $B^\infty y$ vanishes for long horizons. Maybe surprisingly, however, $B^\infty$, does feature prominently in the candidate for the associated long-run optimal portfolio, obtained as the limit of its finite-horizon counterparts \eqref{eq:portfolio_heston_finite}:
$$\widehat{\pi}^\infty= \frac{\mu_S}{\gamma}+\frac{\rho\sigma_Y}{\gamma}B^\infty.$$

Of course, the long-run approach is not very useful from a computational point of view if one first solves the more complicated finite-horizon problem and then passes to the limit in the resulting explicit formulas. Instead, it is preferable to directly adapt the corresponding differential equations. Indeed, with the ansatz \eqref{eq:vlr}, the HJB equation \eqref{eq:HJB2} for the Heston model simplifies to
$$
-A^\infty=\frac{\gamma-1}{\gamma}\left(\left(\frac{\mu_S^2y}{2}+\gamma r\right)+\mu_S \rho \sigma_Y y B^\infty+\frac{\rho^2\sigma_Y^2}{2} (B^\infty)^2 \right)-B^\infty\lambda_Y(\bar {Y}-y)-\frac{\sigma_Y^2}{2}(B^\infty)^2.
$$
Separating the terms proportional to and independent of $y$, this leads to a simple quadratic equation for $B^\infty$ (rather than an ODE with quadratic right-hand side as for the finite-horizon problem), that in turn directly determines $A^\infty$:
\begin{align}
0 &= c(B^\infty)^2+bB^\infty+a,\notag \\
A^\infty &= (1-\gamma)r+\lambda_Y \bar{Y} B^\infty.\label{eq:Ainfty}
\end{align}
Here, the constants $c,b,a$ are defined as in \eqref{eq:abcheston}. Let us focus on the economically more relevant case $\gamma>1$. Then, the candidate equivalent safe rate $A^\infty/(1-\gamma)$ is decreasing in $B^\infty$. Hence, the smaller solution of the quadratic equation is the obvious candidate in this case:\footnote{Note that $c<0$. One then readily verifies that this expression indeed coincides with the long-horizon limit $T\to \infty$ in its finite-horizon counterpart \eqref{eq:Bheston}.}
$$B^\infty=\frac{\sqrt{D}-b}{2c},$$
where the discriminant $D=b^2-4ac$ is assumed to be positive as in the finite horizon case, compare Footnote~\ref{discrimantpos}. This in turn yields: 
$$A^\infty=(1-\gamma)r+ \lambda_Y \bar{Y} \frac{\sqrt{D}-b}{2c}.$$

\paragraph{Small-Noise Asymptotics} Compared to the finite-horizon value function \eqref{eq:affine}, the equivalent safe rate $A^\infty/(1-\gamma)$ is considerably easier to interpret. This is because it encodes the attractiveness of trading in the market into one number, irrespective of current values of the state variable, which are irrelevant in the long-run. To shed more light on the still rather complicated formula above, consider its asymptotics for $\sigma_Y \sim 0$, i.e., the case where the variance process only fluctuates slowly so that the Heston model is close to Black-Scholes.\footnote{A different asymptotic regime, namely portfolio choice with quickly mean-reverting stochastic volatility is studied in \cite{fouque.al.13}.} Taylor expansion yields
\begin{equation}\label{eq:expansion heston}
\frac{A^\infty}{1-\gamma}=r+\frac{\mu_S^2 \bar{Y}}{2\gamma}\left(1+(1-\gamma)\frac{\mu_S}{\gamma}\frac{\rho\sigma_Y}{\lambda_Y}+O(\sigma_Y^2)\right).
\end{equation}
Likewise, the corresponding expansion of the optimal portfolio reads as 
$$\widehat\pi^\infty_t=\frac{\mu_S}{\gamma}\left(1+(1-\gamma)\frac{\mu_S}{\gamma} \frac{\rho\sigma_Y}{2\lambda_Y}+O(\sigma_Y^2)\right).$$
As $\sigma_Y \to 0$, the Heston-type model (\ref{eq:heston1}-\ref{eq:heston2}) converges to a Black-Scholes model with mean return $\mu_S \bar{Y}$ and volatility $\sqrt{\bar{Y}}$. The equivalent safe rate for the latter is $r+\frac{\mu_S^2 \bar{Y}}{2\gamma}$, in line with the zeroth-order term in \eqref{eq:expansion heston}. The sign of the first-order correction for stochastic volatility depends on the investor's risk aversion. For high risk aversion ($\gamma>1$), the equivalent safe rate is increased, and conversely for $\gamma<1$. As explained in Section~\ref{sec:heston1} and~\ref{sec:ko}, the interpretation is that the hedging effect of negative correlation between current returns and future investment opportunities increases welfare for investors with high risk aversion, who welcome this effect. The size of the (relative) first-order adjustment is given by two times of its counterpart for the intertemporal hedging term. Both of these terms are large if i) the model is far from Black-Scholes for large fluctuations or slow mean reversion of the variance process (i.e., large $\sigma_Y/\lambda_Y$) or ii) if the instantaneous correlation $\rho$ between return and volatility shocks is strongly negative.

\paragraph{Predictable Returns}
For the model of Kim and Omberg, one can proceed similarly as for the stochastic volatility model considered above. More specifically, the ansatz
\begin{equation}\label{eq:vfkolr}
V^\infty(t,x,y)=\frac{x^{1-\gamma}}{1-\gamma} \exp\left(A^\infty(T-t)+B^\infty y+\frac{1}{2}C^\infty y^2\right)
\end{equation}
simplifies the reduced HJB equation \eqref{eq:HJB2} to the following system of algebraic equations:
\begin{align}
0 &= c (C^\infty)^2+bC^\infty+a,\label{eq:odeko1}\\
0 &= 2c B^\infty C^\infty+bB^\infty-\lambda_Y \bar{Y} C^\infty,\label{eq:odeko2}\\
A^\infty &= (1-\gamma)r-c(B^\infty)^2+\lambda_Y \bar{Y} B^\infty+\frac{\sigma_Y^2}{2} C^\infty,\label{eq:odeko3}
\end{align}
where the constants $c,b,a$ are defined as in \eqref{eq:abcko} above. Let us again focus on the economically more relevant case $\gamma>1$. Then, the putative equivalent safe rate $A^\infty/(1-\gamma)$ is decreasing in $C^\infty$, so that the smaller solution of the quadratic equation \eqref{eq:odeko1} is the right candidate. As $c<0$, it can be written as:
\begin{equation}\label{eq:Cko}
C^\infty= \frac{\sqrt{D}-b}{2c},
\end{equation}
given that the discriminant $D=b^2-4ac$ is positive.\footnote{As in the Heston model, some algebra shows that this always holds for $\gamma>1$. For $\gamma<1$ this condition is satisfied, e.g., if the volatility $\sigma_Y$ of the factor process is small enough compared to its mean-reversion speed $\lambda_Y$, i.e., if the model is not too far away from Black-Scholes.} Next, \eqref{eq:odeko2} gives
$$B^\infty=\frac{\lambda_Y \bar{Y} C^\infty}{2cC^\infty+b},$$
and $A^\infty$ is in turn determined by \eqref{eq:odeko3}. 

\paragraph{Small-Noise Expansion} To shed more light on these explicit but involved formulas, again consider a small-noise expansion for small fluctuations of the state variable, $\sigma_Y \sim 0$. After some tedious but straightforward algebra (best carried out with a computer algebra system), the expansion for the candidate equivalent safe rate is determined as: 
$$\frac{A^\infty}{1-\gamma}=r+\frac{\bar{Y}^2}{2\gamma\sigma^2}\left(1+(1-\gamma)\frac{2\rho\sigma_Y}{\gamma\lambda_Y\sigma}+O(\sigma_Y^2)\right).$$
The zeroth-order term is again the equivalent safe rate of the Black-Scholes model with constant expected excess return $\bar{Y}$ and volatility $\sigma$ that arises in the limit $\sigma_Y \to 0$. As for the first-order correction due to return predictability, notice that the average (over realizations of the state variable $Y_t$ starting from $Y_0=\bar{Y}$) of the long-run portfolio is given by
$$E[\widehat\pi^{\infty}_t]=\frac{\bar{Y}}{\gamma\sigma^2}\left(1+(1-\gamma)\frac{\rho\sigma_Y}{\gamma\lambda_Y \sigma}+O(\sigma_Y^2)\right).$$
Hence, on average, the first-order (relative) corrections to the equivalent safe rate and the long-run portfolio are again the same up to a factor 2, like in the stochastic volatility model considered above. These terms are large if i) the model is far from Black-Scholes because the state variable is either fluctuating a lot compared to the risky returns (i.e., if $\sigma_Y/\sigma$ is large) or is slowly mean-reverting for small $\lambda_Y$, or ii) if the correlation between return and state shocks are strongly negative. Like for stochastic volatility, the sign of these effects depends on risk aversion. For high risk aversion ($\gamma>1$), welfare is increased, and vice versa for $\gamma<1$. The interpretation is the same as for the stochastic volatility model above.

\section{Verification}\label{sec:verification}

In the previous section, we have \emph{heuristically computed} candidates for optimal portfolios using the dynamic programming approach of stochastic control. To \emph{verify rigorously} that these are indeed optimal, further work is necessary and different methods can be used. Here, we employ tools from convex duality,\footnote{Alternatives include verification theorems (cf., e.g., \cite{korn.97,kraft.05,davis.norman.90}) as well as arguments based on the powerful machinery of viscosity solutions (compare, e.g., \cite{fleming.soner.06,shreve.soner.94,zariphopoulou.01,pham.09}).} that are particularly suited to the long-run version \eqref{eq:lr} of the problem.\footnote{Here, we only consider three benchmark models that can be solved form. Long-run optimality in a general Markovian setting is a studied by Guasoni and Robertson \cite{guasoni.robertson.12}.} 

The basic idea is the following. The utility derived from applying our candidate portfolio evidently yields a \emph{lower} bound for the value function. To determine an \emph{upper} bound, notice that the discounted wealth process $e^{-rt}X^\pi_t$ of any portfolio $\pi_t$ is a local martingale under any equivalent martingale measure, i.e., $e^{-rt}X^\pi_t$ becomes a local martingale when multiplied with the corresponding density process $M_t$. Put differently, any undiscounted wealth process $X^{\pi}_t$ becomes a local martingale when multiplied with $Z_t=e^{-rt}M_t$, where $M_t$ is the density process of any equivalent martingale measure. By Fatou's lemma, the product $X^\pi_t Z_t$ therefore is a supermatingale because it is positive. Hence the following simple consequence of H\"older's inequality (cf.\ \cite[Lemma 5]{guasoni.robertson.12}) bounds the utility of \emph{any} portfolio by a moment of \emph{any} equivalent martingale measure or, more generally, any \emph{supermartingale deflator} $Z_t$ for which $Z_0=1$ and whose product $X^\pi_t Z_t$ with any wealth process is a supermartingale:

\begin{mylemma}\label{lem:dual}
Let $X$, $Z$ be strictly positive random variables such that $E[XZ] \leq x$. Then: 
\begin{equation}\label{eq:duality}
\frac{1}{1-\gamma} E\left[X^{1-\gamma}\right] \leq \frac{x^{1-\gamma}}{1-\gamma} E\left[Z^{1-1/\gamma}\right]^\gamma,
\end{equation}
and equality holds if and only if $E[XZ]=x$ and, for some $\alpha>0$,
\begin{equation}\label{eq:focdual}
X^{-\gamma}=\alpha Z.
\end{equation}
\end{mylemma}

Before recalling the short proof from \cite[Lemma 5]{guasoni.robertson.12}, some remarks are in order here. 

\begin{myrek}
\begin{enumerate}
\item The duality bound \eqref{eq:duality} provides a \emph{sufficient} condition for optimality. Indeed, a payoff $X^{\widehat\pi}_T$ generated by some portfolio is evidently optimal for the finite-horizon problem \eqref{eq:tw} if one can find some supermartingale deflator $\hat{Z}_t$ such that \eqref{eq:duality} holds with equality for the terminal value $\widehat{Z}_T$. For the long-run problem \eqref{eq:lr}, it instead suffices to show that the long-run growth rates of both sides in \eqref{eq:duality} coincide for the wealth process of a candidate portfolio $\widehat{\pi}^\infty_t$ and some supermartingale deflator $\widehat{Z}^\infty_t$:
\begin{equation}\label{eq:duallong}
\liminf_{T \to \infty} \frac{1}{(1-\gamma)T} \log E\left[(X_T^{\widehat{\pi}^\infty})^{1-\gamma}\right] = \liminf_{T \to \infty} \frac{\gamma}{(1-\gamma)T}\log E\left[(\widehat{Z}^\infty_T)^{1-1/\gamma}\right].
\end{equation}
\item With substantial additional effort, one can show that equality in the upper bound \eqref{eq:duality} is also \emph{necessary} for optimality. That is, there always exists a dual variable satisfying the \emph{first-order condition} \eqref{eq:focdual}. This duality between attainable payoffs and supermartingale deflators in fact holds for much more general preferences and asset price dynamics, as shown for complete markets by \cite{pliska.86,karatzas.al.87,cox.huang.89,cox.huang.91} and for incomplete markets by \cite{he.pearson.91a,he.pearson.91b,karatzas.al.91,kramkov.schachermayer.99}. Since this is not needed for our purposes, we do not go into details here but instead refer interested readers to the survey \cite{schachermayer.04} and the references therein.
\item The optimal dual variable $\widehat{Z}_t$ attaining the duality bound \eqref{eq:duality} is intimately linked to ``marginal utility-based prices'' for the investor at hand. Indeed, the investor's utility remains the same if she sells a \emph{small} claim with payoff $H$ at $T$ for a price of $E[\widehat{Z}_T H]$. This ``marginal price'' of a small claim is given by the expectation of its payoff ``deflated'' by the dual variable $\widehat{Z}_T$, resp.\ its ``risk-neutral'' expectation $E^\mathbb{Q}[H]$ if $\widehat{Z}_T$ is the density of an equivalent martingale measure $\mathbb{Q}$. Hence, the investor's ``marginal pricing rule'' for small risks is selected from the possibly infinitely many equivalent martingale measures according to her preferences by the first-order condition \eqref{eq:focdual}. For precise statements and proofs, see \cite{davis.97,karatzas.kou.96,kramkov.sirbu.06}.

\end{enumerate}
\end{myrek}

\begin{proof}[Proof of Lemma \ref{lem:dual}]
To ease notation, set $p=1-\gamma$ and $q=p/(p-1)=1-1/\gamma$. For $p \in (0,1)$, H\"older's inequality (with $\widetilde{p}=1/p$ and $\widetilde{q}=\widetilde{p}/(\widetilde{p}-1)=1/(1-p)$) yields
\begin{align*}
E[X^p] = E[(XZ)^p Z^{-p}] \leq E[(XZ)^{p\widetilde{p}}]^{1/\widetilde{p}} E[Z^{-p\widetilde{q}}]^{1/\widetilde{q}}=E[XZ]^{1/\widetilde{p}}E[Z^q]^{1-p} \leq x^p E[Z^q]^{1-p},
\end{align*}
where we have used $E[XZ] \leq x$ for the last inequality. Dividing by $p>0$, the assertion follows. Similarly, if $p<0$ and in turn $q<1$, H\"older's inequality with $\widetilde{p}=1/(1-q)$ and $\widetilde{q}=\widetilde{p}/(\widetilde{p}-1)=1/q$ yields
$$ E[Z^q]^{1-p} = E[(XZ)^q X^{-q}]^{1-p} \leq E[(XZ)^{q\widetilde{q}}]^{(1-p)/\widetilde{q}} E[X^{-q\widetilde{p}}]^{(1-p)/\widetilde{p}}= E[XZ]^{-p} E[X^p] \leq x^{-p} E[X^p],$$
so that the assertion follows by multiplying with $x^p/p<0$. In both cases, the inequalities become equalities if and only if $E[XZ]=x$ and $X^{-\gamma}$ is proportional to $Z$. 
\end{proof}

\subsection{Constant Investment Opportunities}

In the Black-Scholes model it is straightforward to apply Lemma \ref{lem:dual} to verify the optimality of our candidate portfolio $\widehat\pi=\mu/\gamma\sigma^2$ from Section \ref{sec:BS1}. Since the market is complete, there is only one equivalent martingale measure, so that we only have to show that the bound \eqref{eq:duality} is tight for the latter and the wealth process of our candidate portfolio:

\begin{mylemma}
Fix a time horizon $T>0$. Then, the wealth process $X^{\widehat{\pi}}_T$ of the portfolio $\widehat{\pi}=\mu/\gamma\sigma^2$ and the density 
$$M_T=\exp\left(-\frac{\mu}{\sigma}W_T -\frac{\mu^2}{2\sigma^2}T\right)$$
of the unique equivalent martingale measure satisfy:
\begin{equation}\label{eq:boundsBS}
E\left[(X^{\widehat\pi}_T)^{1-\gamma}\right]= x^{1-\gamma} \exp\left((1-\gamma)\left(r+\frac{\mu^2}{2\gamma\sigma^2}\right)T\right) = x^{1-\gamma} E\left[(e^{-rT}M_T)^{1-\frac{1}{\gamma}}\right]^{\gamma}. 
\end{equation}
By Lemma \ref{lem:dual}, the portfolio $\widehat\pi$ is therefore optimal for the finite-horizon problem \eqref{eq:tw} and also long-run optimal in the sense of \eqref{eq:lr}, with equivalent safe rate $r+\frac{\mu^2}{2\gamma\sigma^2}$.
\end{mylemma}

\begin{proof}
By \eqref{eq:stoexp}, the wealth process corresponding to the constant portfolio $\widehat\pi$ is given by
$$X^{\widehat\pi}_T= x \exp\left(\left(r+\widehat\pi \mu-\frac{\widehat\pi^2\sigma^2}{2}\right)T+\widehat\pi \sigma W_T\right)=x \exp\left(\left(r+\left(1-\frac{1}{2\gamma}\right)\frac{\mu^2}{\gamma\sigma^2}\right)T +\frac{\mu}{\gamma\sigma}W_T\right).$$
Hence, both the payoff $X^{\widehat\pi}_T$ and the terminal value $Z_T=e^{-rT}M_T$ of the supermartingale deflator corresponding to the unique equivalent martingale measure are log-normally distributed. \eqref{eq:boundsBS} in turn follows from the formula for the corresponding moment-generating function.
\end{proof}

\subsection{Stochastic Volatility}\label{sec:heston2}

Let us now turn to the Heston-type model (\ref{eq:heston1}-\ref{eq:heston2}). Here, we focus on the long-horizon asymptotics from Section \ref{sec:lr1}; verification theorems for the finite-horizon problem can be found in \cite{kraft.05,kallsen.muhlekarbe.10}. 
\paragraph{Heuristic Derivation of the Dual Deflator}
Compared to the Black-Scholes model a new difficulty arises here: since the Heston model is incomplete, there are infinitely many equivalent martingale measures, so that one needs to identify a suitable candidate to match the duality bound \eqref{eq:duality}. To this end,  notice that the value function (which we determined heuristically in Section \ref{sec:heston1} to identify the optimal portfolio) also encodes the dual martingale measure. Informally, this can be seen as follows. For any small $\varepsilon$, the dynamic programming principle \eqref{eq:dpp} implies that, conditional on the information at time $t$, the optimal policy $\widehat\pi_t$ performs at least as well as shifting an extra amount $\varepsilon$ into the risky account from time $t$ to maturity $T$:
$$E[U(X^{\widehat\pi}_T)|\scr{F}_t] \geq E\left[U\left(\Big(1-\frac{\varepsilon}{X^{\widehat\pi}_t}\Big)X^{\widehat\pi}_T+\frac{\varepsilon}{S_t}S_T\right) \Big| \scr{F}_t \right].$$
Since this holds both for small positive and negative values of $\varepsilon$, a first-order Taylor expansion yields:
\begin{equation}\label{eq:FOC} 
0=E\left[U'(X^{\widehat\pi}_T)\Big(\frac{X^{\widehat\pi}_T}{X^{\widehat\pi}_t}-\frac{S_T}{S_t}\Big)\Big| \scr{F}_t\right].
\end{equation}
Recall that $U'(x)x=(1-\gamma)U(x)=(1-\gamma)V(T,x,y)$, and notice that the arguments of Section \ref{sec:hjb} show that $V(t,X^{\widehat\pi}_t,Y_t)$ is a martingale for the optimizer $\widehat\pi_t$.\footnote{To see this, notice that the dynamic programming equation \eqref{eq:HJB1a} implies that the drift terms in \eqref{eq:dynv} vanish for the optimizer $\widehat\pi_t$.} As a result: 
$$E\left[U'(X^{\widehat\pi}_T)\frac{X^{\widehat\pi}_T}{X^{\widehat\pi}_t}\Big|\scr{F}_t\right]=\frac{(1-\gamma)E[V(T,X^{\widehat\pi}_T,Y_T)|\scr{F}_t]}{X^{\widehat\pi}_t}=\frac{(1-\gamma)V(t,X^{\widehat\pi}_t,Y_t)}{X^{\widehat\pi}_t}=V_x(t,X^{\widehat\pi}_t,Y_t),$$
where we have used the homotheticity \eqref{eq:hom} of the value function in the last step. Hence, \eqref{eq:FOC} implies  
$$E\left[V_x(T,X^{\widehat\pi}_T,Y_T) \frac{S_T}{S_t} \Big| \scr{F}_t\right]= V_x(t,X^{\widehat\pi}_t,Y_t) \frac{S_t}{S_t},$$
so that $V_x(t,X^{\widehat\pi}_t,Y_t)S_t$ must be a martingale. Repeating these arguments, but now moving $\varepsilon$ into the safe account at time $t$, it follows analogously that $V_x(t,X^{\widehat\pi}_t,Y_t)S^0_t$ has to be a martingale, too. Hence, $V_x(t,X^{\widehat\pi}_t,Y_t)S^0_t$ should be the density process of an equivalent martingale measure after normalizing the initial value to $1$, and $V_x(t,X^{\widehat\pi}_t,Y_t)$ is the corresponding deflator that turns all wealth processes into supermartingales. Since $V_x(T,X^{\widehat\pi}_T,Y_T)=U'(X^{\widehat\pi}_T)$ by the terminal condition \eqref{eq:terminal1}, Lemma \ref{lem:dual} therefore suggests that the upper duality bound \eqref{eq:duality} is attained by the wealth derivative of the value function, evaluated along the wealth process of the candidate portfolio, and normalized to initial value $1$.

The corresponding candidate deflator in the long-run limit $T \to \infty$ is derived analogously from the candidate long-run value function \eqref{eq:vlr}: differentiate with respect to wealth and normalize the initial value to $1$, obtaining
\begin{equation}\label{eq:Zinftytemp}
\widehat{Z}^\infty_t:=(X^{\widehat{\pi}}_t/x)^{-\gamma}\exp(-A^\infty t +B^\infty (Y_t-Y_0)).
\end{equation}

\paragraph{Verification} With candidates for the optimal portfolio and the dual deflator at hand, we can now verify \emph{rigorously} that they attain the bound \eqref{eq:duality}. The first step is to show that the constant $B^\infty$ -- which determines the intertemporal hedging term $\frac{\rho \sigma_Y}{\gamma} B^\infty$ and, together with \eqref{eq:Ainfty}, the density~\eqref{eq:Zinftytemp} of the candidate for the dual deflator -- is well defined:

\begin{mylemma}\label{lemmaBinfity}
Suppose that $\gamma>1$, or $\gamma<1$ and
\begin{align}
D=b^2-4 a c &> 0,  \label{determinantcond}
\end{align}
where
\begin{equation}
c=\left(\frac{\gamma-1}{\gamma}\rho^2-1\right)\frac{\sigma_Y^2}{2}, \quad b=\frac{\gamma-1}{\gamma} \mu_S\rho\sigma_Y+\lambda_Y, \quad a=\frac{\gamma-1}{\gamma} \frac{\mu_S^2}{2}.
\end{equation}
Then, there exists a solution $B^\infty \in \mathbb{R}$ of the quadratic equation
\begin{eqnarray}
0 &=& c (B^\infty)^2+ b B^\infty + a \label{odeBinfinity}
\end{eqnarray}
satisfying $B^\infty <0$ for $\gamma>1$ resp.\ $B^\infty > 0$ for $\gamma  < 1$. Moreover:
\begin{eqnarray}\label{lambdaphatpos}
B^\infty <\frac{\lambda_Y- \frac{1-\gamma}{\gamma}\mu_S \rho \sigma_Y}{\sigma_Y^2(1+\frac{1-\gamma}{\gamma}\rho^2)}.
\end{eqnarray}
\end{mylemma}

\begin{proof}
Evidently,
\begin{equation}
B^\infty= \frac{\sqrt{b^2-4ac}-b}{2c}
\end{equation}
solves the quadratic equation \eqref{odeBinfinity}.\footnote{Note that $D$ is always strictly positive for $\gamma>1$, so that the parameter restriction \eqref{determinantcond} is not needed in this case.} Let $\gamma > 1$ and define
\begin{eqnarray*}
f(x) &=& c x^2+ b x + a.
\end{eqnarray*} 
The parabola $f(x)$ opens downwards, and satisfies $f(0)= a > 0$. Taking into account $c<0$, it therefore follows that the smaller root $B^\infty$ of $f(x)$ is negative. For $\gamma <1$, $\rho<0$ implies $f'(0)=b>0$, so that $B^\infty > 0$ follows from $f(0) < 0$. Finally, \eqref{determinantcond} yields
\begin{eqnarray*}
 B^\infty &<& -\frac{b}{2c} =  \frac{\lambda_Y-\frac{1-\gamma}{\gamma}\mu_S \sigma_Y \rho}{\sigma_Y^2\left(1+\frac{1-\gamma}{\gamma}\rho^2\right)},
\end{eqnarray*}
which proves the last part of the assertion.
\end{proof}

As motivated by the heuristics \eqref{eq:Ainfty}, set 
\begin{equation}\label{eq:Ainfty2}
A^\infty = (1-\gamma)r + \lambda_Y \bar Y B^\infty.
\end{equation}

We now verify that the candidate \eqref{eq:Zinftytemp} indeed meets the requirements of Lemma~\ref{lem:dual}:

\begin{mylemma}\label{probabilitymeasureBinfty}
For $B^\infty$ as in Lemma~\ref{lemmaBinfity} and $A^\infty$ as in \eqref{eq:Ainfty2}, in accordance with \eqref{eq:Zinftytemp} define 
$$\widehat{Z}^\infty_t=(X_t^{\widehat{\pi}^\infty}/x)^{-\gamma}\exp(-A^\infty t +B^\infty (Y_t-Y_0)),$$
where $X_t^{\widehat{\pi}^\infty}$ denotes the wealth process of the portfolio $\widehat{\pi}^\infty=(\mu_S+\rho \sigma_Y B^\infty)/\gamma$. Then, 
\begin{equation}\label{eq:Zinfty}
\widehat{Z}^\infty_t = e^{-rt}\scr{E}\left(\int_0^\cdot \sigma_Y \sqrt{Y_s} B^{\infty} d W_s^Y- \int_0^\cdot \sqrt{Y_s}(\mu_S+ \rho \sigma_Y B^
{\infty}) dW_s \right)_t,
\end{equation}
and the process $\widehat{Z}^\infty_t$ satisfies the conditions of Lemma~\ref{lem:dual}, i.e.,
\begin{equation}\label{eq:hestonbound}
E[X^{\pi}_T \widehat{Z}^\infty_T] \leq x,
\end{equation}
where $X^{\pi}_t$ denotes the wealth process of an arbitrary portfolio $\pi_t$.
\end{mylemma}

\begin{proof}
To see this, compute the dynamics of $X^{\pi}_t Z^\infty_t$ using the definition of the stochastic exponential,
\begin{equation}\label{eq:DD}
\scr{E}(N)_t=\exp\left(N_t-N_0-\frac{1}{2}\langle N,N\rangle_t\right),
\end{equation}
as well as Yor's formula,
\begin{equation}\label{eq:Yor}
\scr{E}(N)_t\scr{E}(N')_t=\scr{E}(N+N'+\langle N,N'\rangle)_t,
\end{equation}
for continuous semimartingales $N_t, N'_t$. The drift rate of $X^{\pi}_t Z^\infty_t$ turns out to be zero for any portfolio $\pi_t$, so that this process is a nonnegative local martingale and in turn a supermartingale as claimed. For the convenience of the reader, we provide some details of the somewhat lengthy but straightforward calculations needed to check this. By \eqref{eq:stoexp} and \eqref{eq:DD}, we have
\begin{eqnarray*}
(X_t^{\widehat{\pi}_\infty})^{-\gamma} &=& x^{-\gamma}\exp{\left(\int_0^t -\gamma r - \mu_S Y_s (\mu_S+ \rho \sigma_Y B^{\infty}) ds+ \int_0^t \frac{1+\gamma}{2\gamma} Y_s (\mu_S+ \rho \sigma_Y B^{\infty})^2  ds\right)}\\
 & &\times \scr{E}\left(\int_0^\cdot - \sqrt{Y_s}(\mu_S+ \rho \sigma_Y B^{\infty}) d W_s\right)_t.
\end{eqnarray*}
Likewise, \eqref{eq:DD} and the dynamics \eqref{eq:heston2} of the process $Y_t$ yield
\begin{eqnarray*}
\exp{\left(-A^\infty t + B^\infty Y_t\right)} &=& \exp{\left(\int_0^t -A^\infty ds + B^{\infty} Y_0 +\int_0^t B^\infty \lambda_Y (\bar{Y}- Y_s) ds + \int_0^t \frac{1}{2} (B^\infty)^2 \sigma_Y^2 Y_s ds \right)}\\
 & & \times \scr{E}\left(\int_0^\cdot \sigma_Y B^\infty \sqrt{Y_s} dW_s^Y\right)_t.
\end{eqnarray*}
Hence, Yor's formula \eqref{eq:Yor} shows
\begin{eqnarray*}
(X_t^{\widehat{\pi}^\infty})^{-\gamma} \exp{\left(-A^\infty t +B^\infty Y_t\right)}
 &=& x^{-\gamma} \exp{\left(B^\infty Y_0\right)} \exp{\left(\int_0^t \left[-A^\infty -\gamma r + B^\infty \lambda_Y (\bar{Y}- Y_s) \right] ds\right)} \\
  & & \times \exp{\left(\int_0^t \left[\frac{1}{2}(B^\infty)^2 \sigma_Y^2 Y_s+\frac{1-\gamma}{2 \gamma} Y_s (\mu_S+ \rho \sigma_Y B^\infty)^2\right] ds\right)}\\
 & & \times \scr{E}\left(\int_0^\cdot \sigma_Y \sqrt{Y_s} B^\infty d W_s^Y- \int_0^\cdot \sqrt{Y_s}(\mu_S+ \rho \sigma_Y B^\infty) dW_s \right)_t\\
 &=& x^{-\gamma} \exp{\left(-rt+B^\infty Y_0\right)} \\
 & & \times \scr{E}\left(\int_0^\cdot \sigma_Y \sqrt{Y_s} B^\infty d W_s^Y- \int_0^\cdot \sqrt{Y_s}(\mu_S+ \rho \sigma_Y B^\infty) dW_s \right)_t,
\end{eqnarray*}
where we have used the quadratic equation~\eqref{odeBinfinity} for $B^\infty$ in the last step. This proves the first part \eqref{eq:Zinfty} of the assertion. Now, let $\pi_t$ be any portfolio. Then, \eqref{eq:stoexp}, \eqref{eq:Zinfty} and Yor's formula \eqref{eq:Yor} give
\begin{eqnarray*}
E\left[X^{\pi}_T \widehat{Z}^\infty_T\right] &=& E \Biggl[x\scr{E}\left(\int_0^\cdot \pi_t \mu_S Y_t dt + \pi_t \sqrt{Y_t} dW_t\right)_T\\
 & & \qquad \times \scr{E}\left(\int_0^\cdot \sigma_Y \sqrt{Y_t} B^{\infty} d W_t^Y- \int_0^\cdot \sqrt{Y_t}(\mu_S+ \rho \sigma_Y B^{\infty}) dW_t \right)_T\Biggr]\\
  &=& E \left[x\scr{E}\left(\int_0^\cdot \sigma_Y \sqrt{Y_t} B^{\infty} d W_s^Y+ \int_0^\cdot \sqrt{Y_t} \left[\pi_t-(\mu_S+ \rho \sigma_Y B^{\infty})\right]d W_t\right)_T\right] \\
  &\leq & x,
\end{eqnarray*}
because a positive local martingale is a supermartingale by Fatou's Lemma.
\end{proof}

To establish the long-run optimality of the portfolio $\widehat{\pi}^\infty$, we now compute upper and lower finite horizon bounds. Here, the candidate portfolio itself leads to the lower bound, whereas the upper bound is derived from the duality bound \eqref{eq:duality}, applied to the supermartingale deflator $\widehat{Z}^\infty_t$ from Lemma \ref{probabilitymeasureBinfty}. In each case, the calculations only make use of the HJB equation, and in turn extend to general Markovian settings, see \cite[Theorem 7]{guasoni.robertson.12} for more details. As observed by Guasoni and Robertson \cite{guasoni.robertson.12}, both bounds are most conveniently expressed in terms of a particular measure $\widehat{\mathbb{P}}$ (locally) equivalent to the physical probability $\mathbb{P}$:\footnote{This measure can be interpreted as the \emph{myopic probability} $\widehat{\mathbb{P}}$, under which a hypothetical investor with logarithmic utility would choose the same long-run optimal portfolio as the original power investor under the physical probability $\mathbb{P}$; see \cite[Section 3.1]{guasoni.robertson.12} for more details. The density process \eqref{eq:myopic} is readily determined by setting the log-optimal portfolio under this measure -- i.e., the corresponding expected excess return determined by Girsanov's theorem, divided by the squared volatility (cf., e.g., \cite{goll.kallsen.00}) -- equal to the candidate long-run portfolio $\widehat{\pi}^\infty$ under the original probability. 

}

\begin{mylemma}\label{lemmafinitehorizon}
Fix a horizon $T>0$ and let $B^\infty$ and $A^\infty$ be defined as in Lemma~\ref{lemmaBinfity} and~\eqref{eq:Ainfty2}, respectively. Then, the wealth process $X_t^{\widehat{\pi}^\infty}$ corresponding to the portfolio $\pi^{\infty} = (\mu_S+ \rho \sigma_Y B^\infty)/\gamma$ and the deflator $\widehat{Z}^\infty_t$ defined in Lemma~\ref{probabilitymeasureBinfty} satisfy the following finite-horizon bounds:
\begin{eqnarray}
E\left[(X^{\widehat{\pi}^\infty}_T)^{1-\gamma}\right]&=& x^{1-\gamma} e^{A^\infty T} E^{\widehat{\mathbb{P}}}\left[e^{B^\infty(Y_0-Y_T)}\right],\label{bound1infty}\\
E \left[(\widehat{Z}^\infty_T)^{1-\frac{1}{\gamma}}\right]^{\gamma} &=& e^{A^\infty T} E^{\widehat{\mathbb{P}}}\left[e^{\frac{B^\infty(Y_0-Y_T)}{\gamma}}\right]^{\gamma},\label{bound2infty}
\end{eqnarray}
where the probability measure $\widehat{\mathbb{P}}$ is defined by, 
\begin{eqnarray}\label{eq:myopic}
\frac{d \widehat{\mathbb{P}}|_{\scr{F}_T}}{d \mathbb{P}|_{\scr{F}_T}} = \scr{E}\left(\int_0^\cdot \sqrt{Y_s} \sigma_Y B^\infty d W_s^Y+ \int_0^\cdot \left(\frac{1}{\gamma}-1\right) \sqrt{Y_s} (\mu_S+ \rho \sigma_Y B^\infty) d W_s\right)_T.
\end{eqnarray}
\end{mylemma}

\begin{proof}
To ease notation, set $C^{\infty}= \rho \sigma_Y B^{\infty}$. Let us first verify that the stochastic exponential on the right-hand side of \eqref{eq:myopic} is a true martingale and therefore indeed the density process of $\widehat{\mathbb{P}}$ with respect to $\mathbb{P}$. Since the variance process $Y_t$ does not have finite exponential moments of all orders (cf.\ the proof of Lemma \ref{expboundnovikov} below), Novikov's condition is not directly applicable in its usual form \cite[Corollary 3.5.13]{karatzas.shreve.91}. However, one can apply a version \cite[Corollary 3.5.14]{karatzas.shreve.91} where it is used successively on a sufficiently fine partition of any given interval.\footnote{Instead, one can also turn to results in more general settings warranting the true martingale property of local martingales, cf., e.g., \cite{cheridito.al.05,kallsen.muhlekarbe.10b} and the references therein.} Indeed, L\'evy's theorem implies that
\begin{equation}
W_s^* := \frac{1}{L}(\sigma_Y B^\infty W_s^Y+(1/\gamma-1)(\mu_S+ \rho \sigma_Y B^\infty) W_s)
\end{equation}
with 
\begin{equation*}
L:= \sqrt{\sigma_Y^2 (B^\infty)^2+ 2 C^\infty (1/\gamma-1)(\mu_S+ C^\infty)+ (1/\gamma-1)^2(\mu_S+ C^\infty)^2}
\end{equation*}
is a $\mathbb{P}$-Brownian motion. Now, consider an equidistant partition $0 = t_0< t_1< \cdots < t_N = T$ of $[0,T]$. Lemma~\ref{expboundnovikov} below shows that for $N$ sufficiently large, i.e., $t_i-t_{i-1}$ sufficiently small,
\begin{equation*}
E\left[\exp{\left(\frac{L^2 }{2} \int_{t_k}^{t_{k+1}}  Y_s ds\right)}\right]< \infty, \quad \forall k\geq 0.
\end{equation*}
This in turn implies that $\scr{E}\left(\int_0^\cdot L \sqrt{Y_s} dW_s^*\right)$ is indeed a true martingale by Novikov's Condition as in \cite[Corollary 3.5.14]{karatzas.shreve.91}, and hence the density process of $\widehat{\mathbb{P}}$ with respect to $\mathbb{P}$.
 
Let us now turn to the first finite-horizon bound \eqref{bound1infty}. By \eqref{eq:stoexp} and \eqref{eq:DD}, the wealth process $X_t^{\pi^{\infty}}$ satisfies:
\begin{eqnarray*}
(X^{\widehat{\pi}^\infty}_T)^{1-\gamma}&=& x^{1-\gamma} \exp{\left(\int_0^T\left[(1-\gamma) r+ \frac{1-\gamma}{\gamma} \mu_S Y_t (\mu_S + C^\infty)-\frac{1-\gamma}{2 \gamma^2} Y_t (\mu_S+ C^\infty)^2\right] dt\right)}\\
 & & \times \exp{\left(\int_0^T \frac{1-\gamma}{\gamma} \sqrt{Y_t}(\mu_S + C^\infty) d W_t\right)}.
\end{eqnarray*}
Hence,
\begin{eqnarray}
(X^{\widehat{\pi}^\infty}_T)^{1-\gamma} &=& x^{1-\gamma} e^{rT}Z^\infty_T \exp{\left(\int_0^T \frac{1}{\gamma} \sqrt{Y_t} (\mu_S+ C^\infty) d W_t - \int_0^T \sigma_Y \sqrt{Y_t} B^\infty d W_t^Y\right)}\label{eq:zfirstinfty}\\
 & & \times \exp{\left(\int_0^T \left[(1-\gamma)r + \frac{1}{2} Y_t (B^\infty)^2 \sigma_Y^2 - Y_t C^\infty (\mu_S+ C^\infty) + \frac{1}{2} Y_t (\mu_S + C^\infty)^2 \right] dt \right)}\nonumber\\
  & & \times \exp{\left(\int_0^T \left[\left(\frac{1}{\gamma}-1\right) \mu_S Y_t (\mu_S+ C^\infty)-\frac{1}{2}\left(\frac{1}{\gamma^2}-\frac{1}{\gamma}\right)Y_t (\mu_S+ C^\infty)^2 \right] dt\right)}.\nonumber
\end{eqnarray}
Now, substitute the stochastic integral with respect to $W^Y_t$ in~\eqref{eq:zfirstinfty} using the dynamics \eqref{eq:heston2} of the variance process $Y_t$ and simplify, obtaining 
\begin{eqnarray*}
(X^{\widehat{\pi}^\infty}_T)^{1-\gamma} &=& x^{1-\gamma} e^{rT}Z^\infty_T \exp{\left(\int_0^T \frac{1}{\gamma} \sqrt{Y_t} (\mu_S+ C^\infty)d W_t-\frac{1}{2}\int_0^T \frac{Y_t (\mu_S+ C^\infty)^2}{\gamma^2} dt \right)}\nonumber\\
 & &\times \exp{\left(B^\infty(Y_0-Y_t) + \int_0^T  \left[B^\infty \lambda_Y(\bar{Y}- Y_t)+ \frac{1}{2} (B^\infty)^2 \sigma_Y^2 Y_t + \frac{1-\gamma}{2\gamma}Y_t (\mu_S+ C^\infty)^2\right] dt\right)}\nonumber\\
 & & \times \exp{\left(\int_0^T \left[\frac{\mu_S}{\gamma} Y_t (\mu_S+ C^\infty) + (1-\gamma) r\right]dt\right)}.\nonumber
\end{eqnarray*}
Yor's Formula \eqref{eq:Yor} in turn shows 
\begin{eqnarray*}
\frac{d \widehat{\mathbb{P}}|_{\scr{F}_T}}{d \mathbb{P}|_{\scr{F}_T}} &=& e^{rT} Z^\infty_T \scr{E} \left(\int_0^\cdot \frac{1}{\gamma} \sqrt{Y_t} (\mu_S + C^\infty) d W_t\right) \exp{\left(\int_0^\cdot \frac{1}{\gamma} Y_t \left[(\mu_S + C^\infty)^2- C^\infty (\mu_S + C^\infty) \right] dt\right)}.\\
\end{eqnarray*}
Using the quadratic equation~\eqref{odeBinfinity} and switching to the measure $\widehat{\mathbb{P}}$ we have
\begin{eqnarray}
(X^{\widehat{\pi}^\infty}_T)^{1-\gamma} &=& x^{1-\gamma} \frac{d \widehat{\mathbb{P}}|_{\scr{F}_T}}{d \mathbb{P}|_{\scr{F}_T}} \exp{\left(B^\infty (Y_0-Y_T)+ A^\infty T\right)}.
\end{eqnarray}
The first bound now follows by taking expectations on both sides. The argument for the second bound is similar. By definition,
\begin{eqnarray*}
(\widehat{Z}^\infty_T)^{1-\frac{1}{\gamma}} &=& e^{-rT} e^{rT} \widehat{Z}^\infty_T \exp{\left(\frac{r}{\gamma}T+\int_0^T -\frac{1}{\gamma} \sqrt{Y_t} \sigma_Y B^\infty dW_t^Y + \int_0^T \frac{1}{\gamma}\sqrt{Y_t}(\mu_S+ C^\infty) d W_t\right)}\\
   & & \times \exp{\left(\int_0^T \frac{1}{2 \gamma} \sigma_Y^2 (B^\infty)^2 Y_t dt+ \int_0^T \frac{1}{2 \gamma}(\mu_S + C^\infty)^2 Y_t dt-\int_0^T \frac{1}{\gamma} C^\infty (\mu_S+C^\infty) Y_t dt\right)}.
\end{eqnarray*}
Again replacing the stochastic integral with respect to $W^Y_t$ with \eqref{eq:heston2} and switching to the measure $\widehat{\mathbb{P}}$, we have
\begin{eqnarray*}
(\widehat{Z}^\infty_T)^{1-\frac{1}{\gamma}} &=&  \frac{d \widehat{\mathbb{P}}}{d \mathbb{P}} \exp{\left(\frac{B^\infty(Y_0-Y_T)}{\gamma}+ \frac{A^\infty T}{\gamma}\right)}.
\end{eqnarray*}
The second bound now follows by taking expectations on both sides, raising to the power $\gamma$, and inserting the definition \eqref{eq:Ainfty2} of $A_\infty$.
\end{proof}

The following estimate for the square-root process $Y_t$ was used in the above proof to infer the martingale property of the stochastic exponential in \eqref{eq:myopic}:
 
\begin{mylemma}\label{expboundnovikov}
Let $\alpha$ and $t_2 > 0$. Then for all $0\leq t_1 < t_2$ with $t_2-t_1$ small enough,
\begin{equation*}
E\left[\exp{\left(\alpha \int_{t_1}^{t_2} Y_s ds\right)}\right] < \infty.
\end{equation*}
\end{mylemma}

\begin{proof}
Jensen's inequality and Tonelli's Theorem show that
\begin{eqnarray*}
E \left[\exp{\left(\alpha \int_{t_1}^{t_2} Y_s ds\right)}\right] &\leq & E \left[\frac{1}{t_2-t_1} \int_{t_1}^{t_2} \exp{\left(\alpha (t_2-t_1) Y_s \right)ds}\right]\\
 &= &  \int_{t_1}^{t_2} E\left[ \frac{\exp{\left(\alpha (t_2-t_1) Y_s \right)}}{t_2-t_1}\right]ds\\
 &\leq & \sup_{t_1 \leq t \leq t_2} E\left[\exp{\left(\alpha (t_2-t_1) Y_t\right)}\right].
\end{eqnarray*}
Given the initial value $Y_0$ we know from  \cite{cox.al.85} that $2 \lambda_Y /(\sigma_Y^2 (1-e^{-\lambda_Y t})) Y_t$ follows a noncentral chi-squared distribution, i.e.,
\begin{equation*} 
\frac{2 \lambda_Y}{\sigma_Y^2 (1-e^{-\lambda_Y t})} Y_t \sim \chi^2
\left(\frac{4 \lambda_Y \bar{Y}}{\sigma_Y^2},Y_0 \frac{4 \lambda_Y}{\sigma_Y^ 2}  \frac{e^{-\lambda_Y t}}{1-e^{-\lambda_Y t}}\right).
\end{equation*}
Thus, the moment generating function of the noncentral chi-squared distribution implies that
\begin{eqnarray*}
\sup_{0< t \leq t_2}E \left[\exp{\left(\alpha t_2 Y_t\right)}\right] &=& \sup_{0< t \leq t_2} E \left[\exp{\left(\alpha t_2 \frac{\sigma_Y^2 (1-e^{-\lambda_Y t})}{2 \lambda_Y} \frac{2 \lambda_Y}{\sigma_Y^2 (1-e^{-\lambda_Y t})} Y_t \right)}\right],\\
 &=& \sup_{0< t \leq t_2} \frac{1}{1-\alpha t_2 \frac{\sigma_Y^2}{\lambda_Y} (1-e^{-\lambda_Y t})} \exp{\left(\frac{2 Y_0 \alpha t_2 e^{-\lambda_Y t}}{1-\alpha t_2 \frac{\sigma_Y^2}{\lambda_Y}(1-e^{-\lambda_Y t})}\right)},
\end{eqnarray*}
which is finite for $t_2 < \lambda_Y /(\alpha \sigma_Y^2)$. The assertion for $t_1>0$ follows analogously.
\end{proof}

With the finite horizon bounds (\ref{bound1infty}-\ref{bound2infty}) at hand, we can now apply Lemma~\ref{lem:dual} to establish the long-run optimality of the candidate portfolio $\widehat{\pi}^\infty$:
\begin{mythm}\label{lemmaoptimalityinf}
Suppose that \eqref{determinantcond} is satisfied if $\gamma <1$, or that 
\begin{equation}\label{assumptionbinfgammageq1}
\left(1-2 \frac{\gamma-1}{\gamma}\rho^2\right) \sqrt{b^2-4 ac}+ b > 0, \quad \text{if } \gamma >1,
\end{equation}
where $a,b$ and $c$ are as defined in Lemma~\ref{lemmaBinfity}.\footnote{Recall from Footnote~\ref{discrimantpos} that $D=b^2-4ac$ is always positive for $\gamma>1$.}  Then, in both cases, the candidate portfolio $\widehat{\pi}^\infty = (\mu_S+ \rho \sigma_Y B^\infty)/\gamma$ is long-run optimal and the corresponding maximal equivalent safe rate is given by
\begin{equation*}
\lim_{T\rightarrow \infty} \frac{1}{(1-\gamma) T} \log{E[(X_T^{\widehat{\pi}^\infty})^{1-\gamma}]} =\frac{A^\infty}{1-\gamma}.
\end{equation*}
\end{mythm}

\begin{proof}
Lemma~\ref{lem:dual} and~\eqref{eq:hestonbound} imply that the wealth process of any portfolio $\pi$ satisfies
\begin{equation*}
\frac{1}{1-\gamma} E\left[(X_T^{\pi})^{1-\gamma}\right] \leq \frac{x^{1-\gamma}}{1-\gamma} E\left[(Z^\infty_T)^{1-1/\gamma}\right]^\gamma.
\end{equation*} 
This inequality in turns yields the following upper bound, valid for any portfolio $\pi_t$:
\begin{equation*}
\lim_{T\rightarrow \infty} \frac{1}{(1-\gamma) T} \log{E[(X_T^{\pi})^{1-\gamma}]} \leq  \lim_{T\rightarrow \infty} \frac{\gamma}{(1-\gamma) T} \log{E[(Z^\infty_T)^{\frac{\gamma-1}{\gamma}}]}.
\end{equation*}
In view of~\eqref{bound2infty}, to show that the upper bound of the growth rate equals $A^\infty/(1-\gamma)$, it remains to check that
\begin{equation*}
\lim_{T\rightarrow \infty} E^{\widehat{\mathbb{P}}}\left[ e^{\frac{B^{\infty}(Y_0-Y_T)}{\gamma}}\right] < \infty,
\end{equation*}
i.e., that $e^{-B^\infty Y_T/\gamma}$ has bounded expectation under $\widehat{\mathbb{P}}$ as $T\rightarrow \infty$.
To this end, first note the dynamics of the variance process $Y_t$ under the measure $\widehat{\mathbb{P}}$:
\begin{equation*}
d Y_t = \Big[\lambda_Y \bar{Y}- \underbrace{\left(\lambda_Y- \sigma_Y^2 B^\infty- \sigma_Y \rho \frac{1-\gamma}{\gamma}(\mu_S+ \rho \sigma_Y B^\infty)\right)}_{=: \lambda^{\widehat{\mathbb{P}}, Y}} Y_t\Big] dt+ \sigma_Y \sqrt{Y_t} d W^{\widehat{\mathbb{P}},Y}_t,
\end{equation*}
where 
\begin{equation*}
W^{\widehat{\mathbb{P}},Y}_t = W^Y_t -\int_0^t B^\infty \sigma_Y \sqrt{Y_s} ds + \int_0^t \rho \sqrt{Y_s} \left(1-\frac{1}{\gamma}\right) (\mu_S+ \rho \sigma_Y B^\infty) ds
\end{equation*}
is a $\widehat{\mathbb{P}}$- Brownian motion by Girsanov's Theorem. Notice that $\lambda^{\widehat{\mathbb{P}}, Y}$ is positive by~\eqref{lambdaphatpos}. The Feller Condition $2 \lambda_Y \bar{Y} > \sigma_Y^2$ implies that the stationary law of the square-root process $Y_t$ is a Gamma distribution with shape parameter $2 \lambda_Y \bar{Y} /\sigma_Y^2$ and scale parameter $\sigma_Y^2/2 \lambda^{\widehat{\mathbb{P}}, Y}$ (cf.~\cite{cox.al.85}), whose density function is given by
\begin{equation*}
\nu(dx) = \frac{x^{\frac{2 \lambda_Y \bar{Y}}{\sigma_Y^2}-1} e^{- \frac{2 \lambda^{\widehat{\mathbb{P}}, Y}}{\sigma_Y^2} x}}{\Gamma\left(\frac{2 \lambda_Y \bar{Y}}{\sigma_Y^2}\right) \left(\frac{\sigma_Y^2}{2 \lambda^{\widehat{\mathbb{P}}, Y}}\right)^{\frac{2 \lambda_Y \bar{Y}}{\sigma_Y^2}}}  dx.
\end{equation*}
The ergodic theorem~\cite[Formula II.35]{borodin.salminen.96} in turn yields
\begin{eqnarray*}
\lim_{T\rightarrow \infty} E^{\widehat{\mathbb{P}}}\left[e^{\frac{ -B^\infty Y_T}{\gamma}}\right] &=& \int_0^\infty e^{\frac{ -B^\infty y}{\gamma}} \nu (dy)\\
 &=& \frac{1}{\Gamma\left(\frac{2 \lambda_Y \bar{Y}}{\sigma_Y^2}\right) \left(\frac{\sigma_Y^2}{2 \lambda^{\widehat{\mathbb{P}}, Y}}\right)^{\frac{2 \lambda_Y \bar{Y}}{\sigma_Y^2}}} \int_0^\infty y^{\frac{2 \lambda_Y \bar{Y}}{\sigma_Y^2}-1} e^{-\frac{2}{\sigma_Y^2}\left(\lambda^{\widehat{\mathbb{P}}, Y} + \frac{B^\infty \sigma_Y^2}{2 \gamma} \right)y}  dy< \infty.\\
\end{eqnarray*}
Here, we have used in the last step that $\lambda^{\widehat{\mathbb{P}}, Y}+ B^\infty \sigma_Y^2/2\gamma$ is positive, which follows from Lemma~\ref{lemmaBinfity} and~\eqref{assumptionbinfgammageq1}. A similar argument using~\eqref{bound1infty} shows that the upper bound is attained by the wealth process corresponding to the portfolio $\widehat{\pi}^\infty$, so that the latter is indeed long-run optimal.
\end{proof}

Let us comment briefly on the conditions \eqref{determinantcond} and \eqref{assumptionbinfgammageq1}. The first one is needed to ensure that the maximal expected utility remains finite for arbitrarily long finite horizons (see \cite[Section 3.1]{kallsen.muhlekarbe.10} for more details), which is clearly necessary to obtain a finite equivalent safe rate. Accordingly, this assumption is satisfied automatically if the investor's utility function is bounded from above ($\gamma>1$). In contrast, the second assumption \eqref{assumptionbinfgammageq1} is not needed for low risk aversion ($\gamma<1$), but becomes active if the latter is sufficiently high.\footnote{This condition is a special case of~\cite[Condition (71)]{guasoni.robertson.12}; in the special case considered here ($\nu_0=0$ in their notation), the first condition in \cite[Condition (71)]{guasoni.robertson.12} always holds true. For the parameter estimates of Pan \cite{pan.02}, Assumption \eqref{assumptionbinfgammageq1} is satisfied for all values of $\gamma$.}
 Here, the intuition is that the long-run optimality of $\widehat{\pi}^\infty$ can fail if this portfolio leads to catastrophic results close to maturity because its intertemporal hedging term differs too much from its finite-horizon counterpart. If losses are limited because the utility function is bounded below ($\gamma<1$), then this effect is washed away in the long-run limit. For sufficiently high risk aversion, however, the candidate portfolio $\widehat{\pi}^\infty$ may not be optimal, see \cite[Proposition 25]{guasoni.robertson.12} for more details. This emphasizes the need for rigorous verification theorems, because heuristic computations alone may lead to wrong results.

As a side product, the finite horizon bounds from Lemma \ref{lemmafinitehorizon} allow to assess the performance of the long-run optimal portfolio on any finite horizon. In Figure \ref{fig:esr_heston}, we compare these bounds to the optimal equivalent safe rate that can be obtained by applying the finite-horizon optimizer $\widehat\pi_t$ from \eqref{eq:portfolio_heston_finite}. Evidently, the performance of the long-run optimizer (dot-dashed in Figure~\ref{fig:esr_heston}) cannot be distinguished from the one of the finite-horizon optimizer (solid in Figure~\ref{fig:esr_heston}) even for short horizons here. Both portfolios achieve the long-run optimal growth rate for horizons as short as two years. Moreover, the latter is quite close to its counterpart in a Black-Scholes model with the same mean returns and volatilities, so that the welfare effect of stochastic volatility turns out to be relatively small, in line with the modest intertemporal hedging terms reported in Figure \ref{fig:portfolio_heston}.

\begin{figure}
\centering
\includegraphics[width=0.8\textwidth]{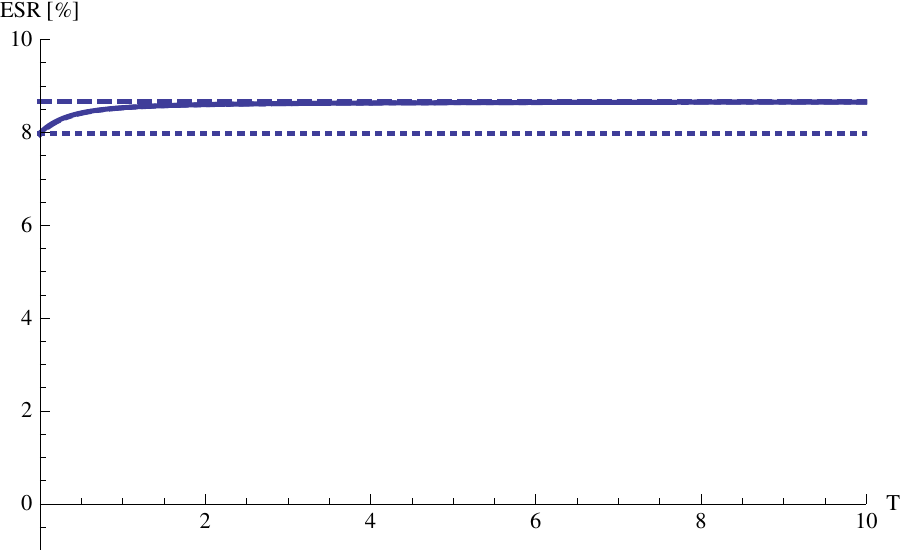}
\caption{\label{fig:esr_heston}
Heston-type model: long-run equivalent safe rate (dashed) and equivalent safe rate for investment on $[0,T]$, for the finite-horizon optimizer $\widehat\pi_t$ on $[0,T]$ (solid), the long-run optimizer $\widehat\pi^\infty$ (dot-dashed), and the optimizer in a Black-Scholes model with the same mean returns and volatilities (dotted). Risk aversion is $\gamma=5$ and the (yearly)  model parameters are $r=0.033$, $\mu_S=4.4$, $\lambda_Y=5.3$, $\bar Y=0.024$, $\sigma_Y=0.38$ and $\rho=-0.57$ (cf.\ \cite[Tables 1 and 6]{pan.02}).
}
\end{figure}

\subsection{Predictable Returns}

Now, let us turn to a rigorous verification theorem for the long-run problem \eqref{eq:lr} in the Kim and Omberg model (\ref{eq:ko1}-\ref{eq:ko2}), where we again focus on the long-run asymptotics. (A verification theorem for the finite horizon problem can be found in \cite{honda.kamimura.11}.) Compared to the Heston-type model (\ref{eq:heston1}-\ref{eq:heston2}) discussed in Section \ref{sec:heston2}, all the formulas become somewhat more involved, but the same basic approach is still applicable. Hence, we provide less details here, focusing on the arguments that require nontrivial adaptations.

As before, the first step is to show that the constant $C^\infty$ -- which determines $B^\infty, A^\infty$ and thereby also the candidates for the intertemporal hedging term and the value function -- is indeed well defined as the root of a quadratic equation, in line with the heuristics leading to \eqref{eq:Cko}:

\begin{mylemma}\label{lemma:GHinftyko}
If $\gamma <1$,\footnote{If $\gamma>1$, one readily verifies that \eqref{determinantcondHko} is always satisfied.} suppose that
\begin{equation}\label{determinantcondHko}
b^2 - 4 a c > 0, 
\end{equation}
where 
\begin{equation*}
c=\left(\frac{\gamma-1}{\gamma}\rho^2-1\right)\frac{\sigma_Y^2}{2}, \quad b=\frac{\gamma-1}{\gamma} \frac{\rho\sigma_Y}{\sigma}+\lambda, \quad a=\frac{\gamma-1}{\gamma} \frac{1}{2\sigma^2}.
\end{equation*}
Then, there exists a solution $C^\infty \in \mathbb{R}$ of 
\begin{eqnarray}
0 &=& c (C^\infty)^2+ b C^\infty + a\label{odeHinfinityko}
\end{eqnarray}
which satisfies $C^\infty < 0$ (for $\gamma>1$) resp.\ $C^\infty > 0$ (for $\gamma <1$). Moreover,
\begin{eqnarray}\label{upperboundBCinfy}
C^\infty <\frac{\lambda_Y-\frac{1-\gamma}{\gamma}\frac{\sigma_Y \rho}{\sigma}}{\sigma_Y^2(1+ \rho^2 \frac{1-\gamma}{\gamma})}.
\end{eqnarray}
\end{mylemma}

\begin{proof}
Define
\begin{eqnarray*}
C^\infty = \frac{\sqrt{b^2-4 a c}-b}{2 c}. 
\end{eqnarray*}
Then, the assertions follow by a similar argument as in Lemma~\ref{lemmaBinfity}.
\end{proof}

In accordance with \eqref{eq:odeko2} and \eqref{eq:odeko3}, set
\begin{align}
B^\infty &= \frac{\lambda_Y \bar{Y} C^\infty}{2 c C^\infty + b},\label{eq:Binfko}\\
A^\infty &= (1-\gamma) r - c (B^\infty)^2+ \lambda_Y \bar Y B^\infty + \frac{\sigma_Y^2}{2}C^\infty.\label{eq:Ainfko}
\end{align}

As in Section \ref{sec:heston2}, a candidate for the long-run dual deflator is readily derived from the candidate long-run value function \eqref{eq:vfkolr} by differentiating with respect to wealth and normalizing. By adapting the argument in Lemma~\ref{lem:dual}, it follows that this candidate is indeed a supermartingale deflator:


\begin{mylemma}\label{probabilitymeasureFGHinftyko}
Let $T> 0$ and define
$$\widehat{Z}^\infty_t=(X_t^{\widehat{\pi}^\infty}/x)^{-\gamma}\exp(-A^\infty t+B^\infty (Y_t-Y_0)+\frac{1}{2}C^\infty (Y_t^2-Y_0^2)),$$
where $X^{\widehat{\pi}^\infty}$ denotes the wealth process of the strategy $\widehat{\pi}^{\infty}_t = \frac{Y_t}{\gamma \sigma^2}+ \frac{\rho \sigma_Y}{\gamma\sigma} (B^\infty+ C^\infty Y_t)$ and the constants $C^\infty$ and $B^\infty$, $A^\infty$ are defined as in Lemma \ref{lemma:GHinftyko} and \eqref{eq:Binfko}, \eqref{eq:Ainfko}, respectively. Then,
\begin{equation*}
\widehat{Z}^\infty_t = e^{-rt}\scr{E}\left(\int_0^\cdot \sigma_Y (B^\infty+ C^\infty Y_s) d W_s^Y- \int_0^\cdot \frac{1}{\sigma}(Y_s+ \rho \sigma \sigma_Y (B^\infty+ C^\infty Y_s)) dW_s \right)_t
\end{equation*}
and the process $\widehat{Z}^\infty_t$ satisfies the conditions of Lemma~\ref{lem:dual}, i.e.,
\begin{equation*}
E \left[X_T^\pi \widehat{Z}^\infty_T\right] \leq x,
\end{equation*}
where $X^\pi_t$ denotes the wealth process of an arbitrary portfolio $\pi_t$.
\end{mylemma}

\begin{proof}
This follows from similar calculations as in Lemma~\ref{probabilitymeasureBinfty}, using the definition of the stochastic exponential \eqref{eq:DD} and Yor's Formula \eqref{eq:Yor}.
\end{proof}

The next step is to compute finite horizon bounds for the wealth process of the candidate portfolio $\widehat{\pi}^\infty_t$ and the corresponding candidate deflator $\widehat{Z}^\infty_t$:

\begin{mylemma}\label{lemmafinitehorizonboundsko}
Fix a time horizon $T>0$ and let $C^\infty$, $B^\infty$ and $A^\infty$ be defined as in Lemma~\ref{lemma:GHinftyko}, \eqref{eq:Binfko}, and \eqref{eq:Ainfko}, respectively. Then, the following finite horizon bounds hold for the wealth process $X^{\widehat{\pi}^\infty}_t$ corresponding to the portfolio $\widehat{\pi}^{\infty}_t = \frac{Y_t}{\gamma \sigma^2}+ \frac{\rho \sigma_Y}{\gamma\sigma} (B^\infty+ C^\infty Y_t)$ and the process $\widehat{Z}^\infty_t$ from Lemma~\ref{probabilitymeasureFGHinftyko}:
\begin{eqnarray}
E\left[(X^{\widehat{\pi}^\infty}_T)^{1-\gamma}\right]&=& x^{1-\gamma} e^{A^\infty T} E^{\widehat{\mathbb{P}}}\left[e^{q(Y_T)-q(Y_0)}\right],\label{bound1infko}\\
E\left[(\widehat{Z}^\infty_T)^{1-\frac{1}{\gamma}}\right]^{\gamma} &=& e^{A^\infty T} E^{\widehat{\mathbb{P}}}\left[e^{\frac{q(Y_T)-q(Y_0)}{\gamma}}\right]^{\gamma},\label{bound2infko}
\end{eqnarray}
where $q(Y_t) = -B^\infty Y_t- C^\infty Y_t^2/2$ and the probability measure $\widehat{\mathbb{P}}$ is defined by
\begin{equation}\label{eq:myopicko}
\frac{d \widehat{\mathbb{P}}|_{\scr{F}_T}}{d \mathbb{P}|_{\scr{F}_T}} = \scr{E} \left(\int_0^\cdot \sigma_Y (B^\infty+ C^\infty Y_s) d W_s^Y+ \int_0^\cdot \left(\frac{1}{\gamma}-1\right) \frac{1}{\sigma} \left(Y_s+ \sigma \sigma_Y \rho (B^\infty + C^\infty Y_s)\right) dW_s \right)_T.
\end{equation}
\end{mylemma}

\begin{proof}
Similarly as in Lemma~\ref{lemmafinitehorizon}, it follows from Lemma~\ref{expboundnovikovko} below, Novikov's Condition as in \cite[Corollary 3.5.14]{karatzas.shreve.91} and the algebraic inequality $(a+b)^2 \leq 2 (a^2+b^2)$ that the stochastic exponential on the right-hand side of \eqref{eq:myopicko} is a true martingale, and therefore is the
 density process of $\widehat{\mathbb{P}}$ with respect to $\mathbb{P}$. The remaining assertions then follow verbatim as in Lemma~\ref{lemmafinitehorizon}.
\end{proof}

To determine the density process of $\widehat{\mathbb{P}}$ with respect to $\mathbb{P}$, the following analogue of Lemma \ref{expboundnovikov} is needed in the proof of Lemma \ref{lemmafinitehorizonboundsko}:

\begin{mylemma}\label{expboundnovikovko}
Let $\alpha$ and $t_2 > 0$. Then for all $0\leq t_1 < t_2$ with $t_2-t_1$ small enough, the Ornstein-Uhlenbeck process $Y_t$ satisfies
\begin{equation*}
E\left[\exp{\left(\alpha \int_{t_1}^{t_2} Y_s^2 ds\right)}\right] < \infty.
\end{equation*}
\end{mylemma}

\begin{proof}
As in Lemma~\ref{expboundnovikov} Jensen's inequality and Tonelli's Theorem show that
\begin{equation}\label{eq:sup}
E \left[\exp{\left(\alpha \int_{t_1}^{t_2} Y_s^2 ds\right)}\right] \leq \sup_{t_1 \leq t \leq t_2} E\left[\exp{\left(\alpha (t_2-t_1) Y_t^2\right)}\right].
\end{equation}
Note that $Y_t$ is Gaussian with $\mathrm{Var}(Y_t) = \frac{\sigma_Y^2}{2 \lambda_Y} (1-e^{-2 \lambda_Y t})\leq \sigma_Y^2/(2 \lambda_Y)$ and that a normally distributed random variable $B_t \sim \mathcal{N}(0,t)$ satisfies $E[\exp{(\xi B_t^2)}]= (1-2\xi t)^{-1/2}$, which is finite for $\xi < 1/2t$. Hence the right-hand side of \eqref{eq:sup} is finite for $t_2-t_1 < \lambda_Y/\alpha \sigma_Y^2$.
\end{proof}

With the finite horizon bounds (\ref{bound1infko}-\ref{bound2infko}) at hand, we can now verify the long-run optimality of the candidate policy $\widehat{\pi}^\infty_t$. The interpretation of the parameter restrictions \eqref{determinantcondHko} for $\gamma<1$ resp.\ \eqref{assumptionbinfgammageq1ko} for $\gamma>1$ is analogous to the discussion for the Heston-type model after Theorem \ref{lemmaoptimalityinf}. For the parameter estimates of Barberis \cite{barberis.00}, Condition
\eqref{assumptionbinfgammageq1ko} is satisfied for risk aversions up to $\gamma=13.4$.

\begin{mythm}\label{lemmaoptinfiniteko}
Suppose that \eqref{determinantcondHko} is satisfied if $\gamma<1$, or that
\begin{equation}\label{assumptionbinfgammageq1ko}
\left(1-2 \frac{\gamma-1}{\gamma} \rho^2\right)\sqrt{b^2-4 ac}+ b>0, \quad \text{if } \gamma >1,
\end{equation}
where $a,b$ and $c$ are defined as in Lemma~\ref{lemma:GHinftyko}. Then, the portfolio $\widehat{\pi}^{\infty}_t = \frac{Y_t}{\gamma\sigma^2}+ \frac{\rho \sigma_Y}{\gamma\sigma} (B^\infty+C^\infty Y_t)$ is long-run optimal and the corresponding maximal equivalent safe rate is given by
\begin{equation*}
\lim_{T\rightarrow \infty} \frac{1}{(1-\gamma) T} \log{E[(X_T^{\widehat{\pi}^\infty})^{1-\gamma}]} =\frac{A^\infty}{1-\gamma}.
\end{equation*}
\end{mythm}

\begin{proof}
Lemma~\ref{lem:dual} and a similar argument as in Lemma~\ref{lemmaoptimalityinf} show that the upper bound equals  $A^\infty /(1-\gamma)$ if $e^{q(Y_T)}$ has bounded expectation under $\widehat{\mathbb{P}}$ as $T \rightarrow \infty$, i.e.,
\begin{equation*}
\lim_{T \rightarrow \infty} E^{\widehat{\mathbb{P}}} \left[e^{\frac{q(Y_T)}{\gamma}}\right] < \infty,
\end{equation*}
with $q$ defined as in Lemma~\ref{lemmafinitehorizonboundsko}. To see that this indeed holds true, first observe that the dynamics of $Y_t$ under the measure $\widehat{\mathbb{P}}$ are given by
\begin{equation*}
d Y_t = \Big[\underbrace{\lambda_Y \bar Y + B^\infty \sigma_Y^2\left(1+\frac{1-\gamma}{\gamma}\rho^2\right)}_{=: L}-\underbrace{\left(\lambda_Y - \frac{\sigma_Y \rho}{\sigma} \frac{1-\gamma}{\gamma}- C^\infty \sigma_Y^2 \left(1+ \frac{1-\gamma}{\gamma}\rho^2\right)\right)}_{=: K}Y_t\Big] dt+ \sigma_Y d W^{\widehat{\mathbb{P}},Y}_t
\end{equation*}
where 
\begin{equation*}
 W^{\widehat{\mathbb{P}},Y}_t= W^Y_t-\int_0^t \sigma_Y (B^\infty + C^\infty Y_s)ds- \int_0^t \frac{\rho}{\sigma} \frac{1-\gamma}{\gamma} (Y_s+ \rho \sigma \sigma_Y(B^\infty+ C^\infty Y_s))ds 
\end{equation*} 
is a $\widehat{\mathbb{P}}$-Brownian motion by Girsanov's theorem. Notice that $K$ is positive by~\eqref{upperboundBCinfy} and recall that the stationary law of the Ornstein-Uhlenbeck process $Y_t$ is Gaussian with mean $L/K$ and variance $\sigma_Y^2/2 K$ (cf., e.g., \cite{borodin.salminen.96}), so that its density function is given by
\begin{eqnarray*}
\nu(dx) &=& \sqrt{\frac{K}{\sigma_Y^ 2 \pi}} e^{- \frac{K}{\sigma_Y^2}\left(x-\frac{K}{L}\right)^2} dx.
\end{eqnarray*}
As a result, the ergodic theorem~\cite[Formula II.35]{borodin.salminen.96} shows that
\begin{eqnarray*}
\lim_{T\rightarrow \infty} E^{\widehat{\mathbb{P}}}\left[e^{\frac{ q(Y_T)}{\gamma}}\right] &=& \int_{-\infty}^\infty e^{-\frac{1}{\gamma}\left(B^\infty y+ C^\infty \frac{y^2}{2}\right)} \nu (d y)\\
 &=&  \sqrt{\frac{K}{\sigma_Y^ 2 \pi}} \int_{-\infty}^\infty  e^{-\frac{B^\infty}{\gamma} y -  \frac{C^\infty}{2\gamma} y^2- \frac{K}{\sigma_Y^2}\left(y-\frac{K}{L}\right)^2}d y < \infty,
\end{eqnarray*}
where we have used in the last step that $\frac{K}{\sigma_Y^2}+\frac{C^\infty}{2 \gamma} >0$, which is a consequence of Lemma~\ref{lemma:GHinftyko} and~\eqref{assumptionbinfgammageq1ko}.

Finally, a similar argument using~\eqref{bound1infko} shows that the upper bound is attained by wealth process corresponding to the portfolio $\widehat{\pi}^{\infty}_t$, so that the latter is indeed long-run optimal.
\end{proof}

 The finite-horizon bounds (\ref{bound1infko}-\ref{bound2infko}) again allow to assess the performance of the long-run optimal portfolio on any finite horizon. In Figure \ref{fig:esr_ko}, the latter is compared to the respective finite-horizons optimizers and to the optimal performance in a Black-Scholes model with the same mean returns and volatilities. In line with the substantial intertemporal hedging terms and the slower convergence reported in Figure \ref{fig:portfolio_ko}, the differences are much more pronounced here than for the stochastic volatility model studied in Section \ref{sec:heston2}. Indeed, for sufficiently long horizons, the hedging opportunities for return predictability allow to achieve considerable welfare gains compared to a Black-Scholes model with the same mean returns and volatilities. For short horizons, however, the long-run portfolio performs badly, as its large intertemporal hedging term leads to a far too risky investment compared to the finite-horizon optimizer in this case. 
 
 \begin{figure}
\centering
\includegraphics[width=0.8\textwidth]{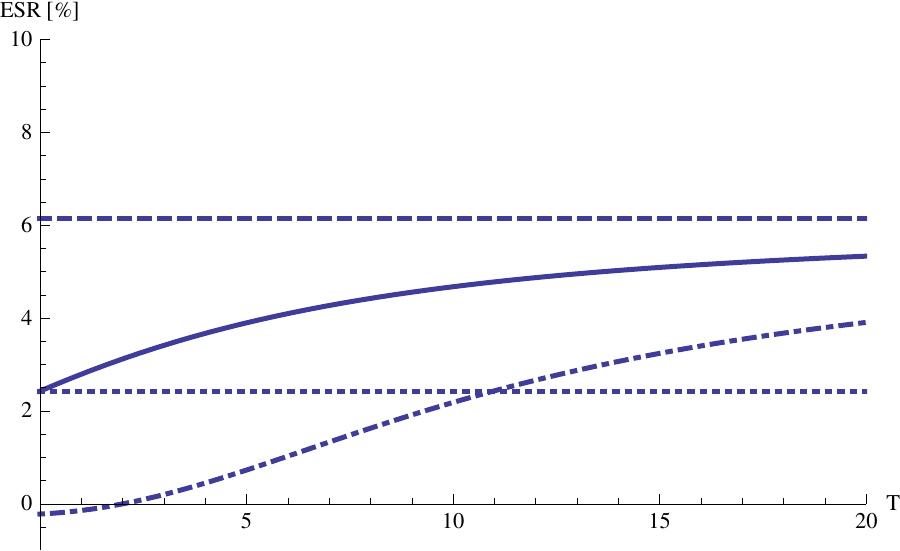}
\caption{\label{fig:esr_ko}
Kim and Omberg model: long-run equivalent safe rate (dashed) and equivalent safe rate for investment on $[0,T]$, for the finite-horizon optimizer $\widehat\pi_t$ on $[0,T]$ (solid), the long-run optimizer $\widehat\pi^\infty$ (dot-dashed), and the optimizer in a Black-Scholes model with the same mean returns and volatilities (dotted). Risk aversion is $\gamma=5$, the (monthly) parameters are $r=0.14\%$, $\sigma=4.36\%$, $\bar{Y}=0.34\%$, $\lambda_Y=2.26\%$, $\sigma_Y=0.08\%$, $\rho=-0.935$ (cf.~\cite{barberis.00,wachter.02}), and
$Y_0=\bar{Y}$.
}
\end{figure}

It is important to note that the substantial benefits from exploiting return predictability are contingent on the simple frictionless model considered here. Parameter uncertainty considerably weakens these results \cite{barberis.00}, and trading costs are also bound to play a key role for the market timing strategies needed to hedge against the future evolution of the state variable. Numerical results on portfolio choice with predictability and transaction costs are reported in \cite{lynch.tan.11}; recent general asymptotic results \cite{martin.12,soner.touzi.13,kallsen.muhlekarbe.13} have also opened the door to explicit results for small costs.

\bibliographystyle{abbrv}
\bibliography{guide}

\end{document}